\documentclass{article}
\usepackage{arxiv}

\usepackage[utf8]{inputenc}
\usepackage[T1]{fontenc}
\usepackage[authoryear]{natbib}
\usepackage[colorlinks,citecolor=blue,urlcolor=blue]{hyperref}
\usepackage{amsmath}
\usepackage{amssymb}
\usepackage{amsthm}
\usepackage{mathtools}
\usepackage{mathrsfs}
\usepackage{graphicx}
\usepackage{booktabs}
\usepackage{enumitem}
\usepackage[capitalise,noabbrev]{cleveref}

\newcommand{\RR}{\mathbb{R}}
\newcommand{\EE}{\mathbb{E}}
\newcommand{\so}{\mathfrak{so}}

\newcommand{\hol}{\mathfrak{hol}}
\DeclareMathOperator{\col}{col}
\DeclareMathOperator{\diag}{diag}
\DeclareMathOperator{\rank}{rank}
\DeclareMathOperator{\tr}{tr}
\DeclareMathOperator{\Sym}{Sym}
\DeclareMathOperator{\sym}{sym}
\DeclareMathOperator{\skewop}{skew}
\DeclareMathOperator{\im}{im}
\DeclareMathOperator{\Hol}{Hol}
\DeclareMathOperator{\Lie}{Lie}
\DeclareMathOperator{\Ad}{Ad}
\DeclareMathOperator{\spn}{span}

\crefname{theorem}{Theorem}{Theorems}
\crefname{proposition}{Proposition}{Propositions}
\crefname{corollary}{Corollary}{Corollaries}
\crefname{conjecture}{Conjecture}{Conjectures}
\crefname{definition}{Definition}{Definitions}
\crefname{equation}{}{}
\crefformat{equation}{(#2#1#3)}
\crefrangeformat{equation}{(#3#1#4) to~(#5#2#6)}
\crefname{figure}{Figure}{Figures}
\crefname{table}{Table}{Tables}
\crefname{section}{Section}{Sections}

\theoremstyle{plain}
\newtheorem{theorem}{Theorem}
\newtheorem{proposition}{Proposition}
\newtheorem{corollary}{Corollary}
\newtheorem{conjecture}{Conjecture}

\theoremstyle{definition}
\newtheorem{definition}{Definition}

\theoremstyle{remark}
\newtheorem*{remark}{Remark}

\title{Random Dot Product Graphs as Dynamical Systems: \\
  Limitations and Opportunities}

\author{Giulio Valentino Dalla Riva \\
  Baffelan.com \\
  \texttt{me@gvdallariva.net}}

\renewcommand{\shorttitle}{RDPG as Dynamical Systems}
\renewcommand{\headeright}{}
\renewcommand{\undertitle}{}

\begin{document}
\maketitle

\begin{abstract}
Across ecology, history, economics, social behavior, cultural dynamics, and more, many phenomena can be described in terms of entities establishing and disestablishing interactions with each other. These scenarios are commonly represented mathematically as temporal networks, and the time evolution of these objects is studied as a time series with the goal of predicting future network states. Here, instead, we take a dynamical systems perspective: when the goal is to understand \emph{why} networks evolve as they do, can we learn the differential equations that govern them?
We investigate this question within the framework of Random Dot Product Graphs (RDPGs), where each network snapshot is generated from latent positions evolving under unknown dynamics. We identify three fundamental obstructions to recovering these dynamics: gauge freedom from the rotational ambiguity in latent positions, realizability constraints from the manifold structure of the probability matrix, and the trajectory recovery problem arising from spectral embedding artifacts.
We develop a geometric framework based on principal fiber bundles that formalizes these obstructions and reveals their interplay. For horizontal families, polynomial dynamics have trivial holonomy, while Laplacian dynamics satisfy an explicit non-commutativity criterion for nontrivial holonomy; in $d = 2$ this yields full restricted holonomy $\mathrm{SO}(2)$, in general $d$ we provide a conditional full-holonomy criterion, and for $d \ge 3$ the generic full $\mathrm{SO}(d)$ statement remains conjectural. Cram\'er--Rao lower bounds show that the spectral gap controlling geometric difficulty simultaneously controls statistical difficulty, an inextricable duality.
We establish an identifiability principle stating that symmetric dynamics cannot absorb skew-symmetric gauge contamination. Yet, we show that significant practical obstacles remain in finite samples. We frame the gap between identifiability and practical recovery as an open challenge and discuss directions for progress.
\end{abstract}


\section{Introduction}\label{sec:intro}

Temporal networks are networks whose edges and nodes change over time. They appear throughout science, from ecological food webs that vary across space and time~\citep{poisot2015species} to neural connectomes that rewire during development.
RDPGs have proven particularly useful for such networks: \citet{dallariva2016exploring} first applied RDPG to ecological networks, showing that latent positions capture evolutionary signatures in food webs; subsequent work has used RDPG embeddings to predict trophic interactions~\citep{strydom2022food}, while the survey by \citet{athreya2017statistical} demonstrates applications to connectome analysis and social network inference.
A fundamental question is: \textbf{what differential equations govern the evolution of network structure?}

This paper investigates that question within the framework of Random Dot Product Graphs (RDPGs)~\citep{athreya2017statistical,young2007random}.
In an RDPG, each node $i$ has a latent position $x_i \in \RR^d$, and the probability of an edge between nodes $i$ and $j$ is $P_{ij} = x_i^\top x_j$.
When these positions evolve according to some dynamics $\dot{X} = f(X)$, the network structure changes accordingly.
Our goal is to recover $f$ from observations.

This approach is appealing for several reasons.
First, the latent space provides a continuous representation where dynamics can be modeled with standard tools from dynamical systems theory.
Second, interpretable dynamics in $X$-space (e.g., attraction, repulsion, diffusion) translate to interpretable changes in network structure.
Third, the RDPG framework has well-developed statistical theory connecting latent positions to spectral properties of observed adjacency matrices.

However, we identify \textbf{three fundamental obstructions} to learning dynamics from RDPG observations:

\begin{enumerate}
\item \textbf{Gauge freedom} (\cref{sec:gauge-freedom}): The latent positions are determined only up to orthogonal transformation (an orthogonal non-identifiability). Indeed, $X$ and $XQ$ produce identical networks for any $Q \in O(d)$. This means some dynamics are \emph{invisible}: they change $X$ without changing observable network structure.

\item \textbf{Realizability constraints} (\cref{sec:realizable}): The probability matrix $P = XX^\top$ lives on a low-dimensional manifold. Not every symmetric perturbation $\dot{P}$ is achievable; those that would increase the rank of $P$ are forbidden.

\item \textbf{Recovering trajectories from embeddings} (\cref{sec:trajectory-problem}): In the RDPG pipeline, we do not observe $X(t)$ or $P(t)$. We observe adjacency matrices and then estimate $X(t)$ via ASE. This estimation step necessarily introduces a time-dependent gauge choice. Even when the underlying $X(t)$ is smooth, the estimated embeddings can jump erratically because the eigendecomposition picks a representative of the $O(d)$ fiber at each time step. This is not something one can fix by taking smaller time steps. It is a structural artifact of eigenvector indeterminacy coupled with estimation from noisy graphs.
\end{enumerate}

We show that existing approaches fail to address these obstructions.
Joint embedding methods like UASE~\citep{gallagher2021spectral} and Omnibus~\citep{levin2017central} assume generative models incompatible with ODE dynamics on latent positions. In fact, they model time-varying ``activity'' against fixed ``identity,'' not genuinely evolving positions.
Bayesian approaches with hierarchical smoothness priors~\citep{loyal2025} produce smooth trajectories but lack \emph{dynamical consistency}: they interpolate well but don't enforce that velocity depends on state according to $\dot{X} = f(X)$.
Pairwise Procrustes alignment is local and cannot ensure global consistency.

A related strategy, pursued by \citet{athreya2024euclidean}, avoids gauge freedom by defining a Procrustes-aligned distance $d_{MV}$ between network states at different times and using classical multidimensional scaling to extract a low-dimensional curve (the ``mirror'') that summarizes network evolution.
Their $d_{MV}$ is gauge-invariant and can be consistently estimated from adjacency spectral embeddings. It is well suited to change-point detection and visualization, but captures the \emph{distance} between network states rather than the \emph{differential equation} governing their evolution, which is the target of the present paper.

Our main contributions are theoretical.
We develop a geometric framework based on principal fiber bundles that formalizes gauge freedom, realizability, and trajectory recovery as distinct but interrelated obstructions.
We prove that random gauge artifacts introduce skew-symmetric contamination that cannot be absorbed by symmetric dynamics (\cref{thm:gauge-contamination}), establishing an identifiability principle.
We establish a concrete contrast among horizontal dynamics families: polynomial dynamics have commuting generators and trivial holonomy, so gauge alignment is purely a statistical problem; for Laplacian dynamics, we prove a local commutator criterion for nontrivial holonomy and a full restricted $\mathrm{SO}(2)$ consequence in dimension $d = 2$, provide a conditional full-holonomy criterion in general $d$, and leave generic full $\mathrm{SO}(d)$ for $d \ge 3$ as a conjecture.
We derive Cram\'er--Rao lower bounds showing that the spectral gap controlling curvature simultaneously controls Fisher information, so geometric and statistical difficulty are inextricable.
However, we show that exploiting these structures in practice faces fundamental difficulties: the bias induced by finite samples and the expressiveness of natural dynamics families make the constructive problem substantially harder than the identifiability theory suggests.
We frame the gap between identifiability and practical recovery as an open problem and discuss directions for progress.

\textbf{Paper outline.}
\Cref{sec:rdpg} reviews RDPG fundamentals.
\Cref{sec:obstructions} develops the geometric framework: gauge freedom, realizability, and the fiber bundle perspective with connections, curvature, and holonomy.
\Cref{sec:trajectory-problem} addresses the practical problem of recovering trajectories from spectral embeddings, showing why existing methods fail and motivating the need for dynamics-aware approaches.
\Cref{sec:dynamics} catalogs concrete dynamics families, analyzes their observable consequences through the Lyapunov equation, and formalizes the polynomial-vs-Laplacian holonomy contrast with explicit proved/conditional/conjectural status.
\Cref{sec:info-theoretic} derives information-theoretic lower bounds revealing the statistical-geometric duality.
\Cref{sec:constructive} discusses the constructive problem: the identifiability principle, its algorithmic implications, the practical obstacles that remain, and a tractable special case using anchor nodes, illustrated with numerical experiments.
\Cref{sec:discussion} reflects on achievements and open problems.

\textbf{Notation.}
$X \in \RR^{n \times d}$ denotes the matrix of latent positions with rows $x_i^\top$.
$\RR_*^{n \times d} = \{X \in \RR^{n \times d} : \rank(X) = d\}$ denotes the set of full-rank $n \times d$ matrices.
$P = XX^\top$ is the probability matrix.
$O(d)$ is the orthogonal group; $\so(d)$ is its Lie algebra of skew-symmetric matrices.
$\hat{X}$ denotes an estimate (e.g., from spectral embedding).
$\|\cdot\|_F$ is the Frobenius norm; $\|M\|_{2 \to \infty} = \max_i \|e_i^\top M\|_2$ is the two-to-infinity norm, measuring the largest row norm of $M$.
$O_p(\cdot)$ denotes stochastic order: $Y_n = O_p(a_n)$ means $Y_n / a_n$ is bounded in probability.
$\Sym(d)$ denotes $d \times d$ symmetric matrices.
Indices $i, j$ refer to nodes; indices $\iota, \gamma$ refer to eigenvector directions when working in a spectral basis.

\textbf{Standing assumptions (default throughout).}
Unless a result states otherwise, we work under:
\begin{enumerate}
\item[\textbf{(S1)}] fixed dimensions with $n \ge d \ge 2$;
\item[\textbf{(S2)}] trajectories on a time interval where $X(t)$ remains in the interior of $\mathcal{E}$ (hence $\rank(X(t)) = d$ and $0 < P_{ij}(t) < 1$);
\item[\textbf{(S3)}] sufficient smoothness of trajectories/fields for the stated derivatives and Lie brackets;
\item[\textbf{(S4)}] for statistical bounds, the regularity assumptions are those listed in \cref{prop:cramer-rao}.
\end{enumerate}

\begin{table}[ht]
\centering
\begin{tabular}{lll}
\toprule
\textbf{Result} & \textbf{Status} & \textbf{Scope} \\
\midrule
\Cref{thm:invisible} & Proved & Full-rank latent states \\
\Cref{prop:curvature-criterion} & Proved & Horizontal symmetric families \\
\Cref{prop:laplacian-holonomy} & Proved & Local nonzero commutator criterion \\
\Cref{cor:linear-fisher} & Proved & Linear dynamics \\
\Cref{prop:cramer-rao} & Proved & Polynomial dynamics + regularity assumptions \\
\Cref{thm:gauge-contamination} & Proved & Symmetric-generator identifiability \\
Full restricted holonomy via curvature span & Conditional theorem & Requires curvature-span hypothesis \\
Finite-time rank criterion at base point & Conditional proposition & Sufficient condition for full restricted holonomy \\
\Cref{conj:full-holonomy} & Conjecture & $d \ge 3$ generic full holonomy \\
\bottomrule
\end{tabular}
\caption{Result status map: proved results vs conditional criteria vs conjecture.}
\label{tab:result-status}
\end{table}

We use \cref{tab:result-status} as the authoritative status legend for theorem claims throughout the manuscript.

\section{Random Dot Product Graphs}\label{sec:rdpg}

We begin with the RDPG framework for undirected networks.
The extension to directed graphs appears in \cref{app:directed}.

\subsection{Latent positions and connection probabilities}

In a Random Dot Product Graph, each node $i \in \{1, \ldots, n\}$ is associated with a latent position $x_i \in \RR^d$.
The probability of an edge between nodes $i$ and $j$ is:
\begin{equation}
P_{ij} = x_i^\top x_j
\end{equation}

For this to be a valid probability, we need $P_{ij} \in [0, 1]$ for all pairs.
A sufficient condition is that all positions lie in the positive orthant of the unit ball:
\begin{equation}
B_+^d = \{x \in \RR^d : x_k \ge 0 \text{ for all } k, \quad \|x\| \le 1\}
\end{equation}
The non-negativity ensures $P_{ij} \ge 0$, and Cauchy--Schwarz gives $P_{ij} \le \|x_i\| \|x_j\| \le 1$.

Collecting positions into a matrix $X \in \RR^{n \times d}$ with rows $x_i^\top$, the probability matrix is:
\begin{equation}
P = XX^\top
\end{equation}
This is symmetric and positive semidefinite with rank at most $d$.
Throughout this paper, we restrict attention to configurations where all edge probabilities lie in the open interval: $0 < P_{ij} < 1$ for all $i, j$.
This excludes boundary cases (deterministic edges or guaranteed non-edges) that would create singularities in the Fisher information and complications in the differential geometry; see \cref{sec:fiber-bundle} for details.

Given latent positions, edges are drawn independently:
\begin{equation}
A_{ij} \sim \text{Bernoulli}(P_{ij}) \quad \text{for } i \le j
\end{equation}
with $A_{ji} = A_{ij}$ (undirected) and $A_{ii} \sim \text{Bernoulli}(P_{ii})$, so $\EE[A] = P$ exactly.
In some applications it is common to disregard self-links by setting $A_{ii} = 0$, yielding $\EE[A] = P - \diag(P)$.
This hollows the diagonal and introduces a systematic off-manifold bias that complicates the spectral theory (the perturbation $\diag(P)$ has spectral norm at most 1, comparable to sampling noise for moderate $n$).
Our geometric and statistical arguments extend to this setting by treating the missing/zeroed diagonal as an additional structured perturbation term; we avoid it in the main text to keep focus and keep $\EE[A]$ exactly on the rank-$d$ manifold $\{XX^\top\}$.
(If desired, the hollow-diagonal extension can be carried out in an appendix by explicitly tracking the extra diagonal perturbation through the embedding and alignment steps.)

\subsection{Adjacency spectral embedding}\label{sec:ase}

Given an observed adjacency matrix $A$, we estimate latent positions via \textbf{adjacency spectral embedding} (ASE)~\citep{athreya2017statistical}.

Since $A$ is symmetric, it has an eigendecomposition $A = U \Lambda U^\top$ with real eigenvalues $\lambda_1 \ge \lambda_2 \ge \cdots \ge \lambda_n$ and orthonormal eigenvectors.
The rank-$d$ ASE is:
\begin{equation}
\hat{X} = U_d |\Lambda_d|^{1/2}
\end{equation}
where $U_d \in \RR^{n \times d}$ contains the $d$ leading eigenvectors (by eigenvalue magnitude) and $\Lambda_d = \diag(\lambda_1, \ldots, \lambda_d)$.
We take absolute values because $A$ can have negative eigenvalues due to sampling noise, even though the true $P$ is positive semidefinite.
In the dense regime where $\lambda_d(P) = \Omega(\sqrt{n})$, Weyl's inequality guarantees that the top $d$ eigenvalues of $A$ are positive with high probability, so the absolute value has no effect asymptotically.
We retain it for definiteness at finite $n$.

\begin{remark}
For symmetric matrices, eigendecomposition and SVD are closely related: if $A = U \Lambda U^\top$ (eigen) and $A = \tilde{U} \Sigma \tilde{V}^\top$ (SVD), then $\tilde{U} = \tilde{V} = U$ (up to signs) and $\sigma_i = |\lambda_i|$.
We use eigendecomposition for ASE (the standard in RDPG literature) but SVD for Procrustes alignment (the standard for that problem).
\end{remark}

\textbf{Statistical properties.}
The statistical properties of RDPG are well studied, and here we remind some fundamental results from the literature. Under mild conditions, ASE is consistent: as $n \to \infty$, the rows of $\hat{X}$ converge to the true latent positions (up to orthogonal transformation) at rate $O(1/\sqrt{n})$~\citep{athreya2016limit,athreya2017statistical}.
Specifically, there exists $Q \in O(d)$ such that $\max_i \|\hat{x}_i - x_i Q\| = O_p(1/\sqrt{n})$, and conditionally on each latent position, $\sqrt{n}(\hat{x}_i - Qx_i)$ converges to a Gaussian whose covariance depends on the edge variance profile~\citep{athreya2016limit,rubindelanchy2022statistical}.
The perturbation expansion $\hat{X} = XW + (A - P)X(X^\top X)^{-1}W + R$, where $W \in O(d)$ and $\|R\|_{2 \to \infty} = o(n^{-1/2})$, shows that the leading error term is linear in $A - P$ and has conditional mean zero~\citep{cape2019two}.
ASE is asymptotically unbiased but not fully efficient for individual latent positions; the one-step procedure of \citet{xie2021efficient} achieves what they term \emph{local efficiency}: for each vertex, the estimator's asymptotic covariance matches that of the oracle MLE that knows all other latent positions, up to an orthogonal transformation.
For stochastic blockmodel parameters, the spectral estimator is itself asymptotically efficient~\citep{tang2022efficient}.
The eigenvalues of $A$ are also asymptotically normal about those of $P$, with an $O(1)$ bias from the Bernoulli variance~\citep{tang2018eigenvalues}.
This bias is negligible relative to the leading eigenvalue ($\lambda_1 \sim n$), but can be comparable to the spectral gap $\delta = \lambda_d$ when the gap is small, and is of the same order as the eigenvalue fluctuations ($O(\sqrt{n})$); for dynamics estimation, what matters is the bias relative to the \emph{changes} in eigenvalues across time steps, a point we revisit in \cref{sec:info-theoretic}.

The orthogonal matrix $Q$ reflects the fundamental gauge ambiguity: eigenvectors are determined only up to sign, and rotations within repeated eigenspaces are arbitrary.
This ambiguity is unavoidable and is the source of the trajectory estimation problem discussed in \cref{sec:trajectory-problem}.


\section{The Gauge Obstruction}\label{sec:obstructions}

We analyze the fundamental obstructions to learning dynamics from RDPG observations, characterizing what is and is not learnable from a theoretical perspective.

\subsection{Gauge freedom and observability}\label{sec:gauge-freedom}

The latent positions $X$ are not uniquely determined by the probability matrix $P$.
For any orthogonal matrix $Q \in O(d)$:
\begin{equation}
(XQ)(XQ)^\top = XQQ^\top X^\top = XX^\top = P
\end{equation}

Thus $X$ and $XQ$ produce identical connection probabilities.
This $O(d)$ symmetry is the \textbf{gauge freedom} of RDPG: the equivalence class $[X] = \{XQ : Q \in O(d)\}$ corresponds to a single observable network structure.

From a geometric point of view, a global rotation of all positions in latent space leaves all pairwise angles and magnitudes unchanged, hence all dot products are preserved.

This gauge freedom has profound implications for learning dynamics.
If $X$ and $XQ$ are observationally indistinguishable, then any dynamics moving along the equivalence class, that is, rotating all positions by a common time-varying orthogonal transformation, produces \textbf{no observable change} in the network.

Consider dynamics $\dot{X} = f(X)$ on the latent positions.
The induced dynamics on the probability matrix $P = XX^\top$ follow from the product rule:
\begin{equation}
\dot{P} = \dot{X} X^\top + X \dot{X}^\top = f(X) X^\top + X f(X)^\top
\end{equation}

\begin{definition}[Observable and Invisible Dynamics]
A vector field $f: \RR^{n \times d} \to \RR^{n \times d}$ produces \textbf{observable dynamics} if $\dot{P} \neq 0$.
Otherwise, the dynamics are \textbf{invisible}: the latent positions change but the network structure remains static.
\end{definition}

The definition calls for a natural question: which dynamics are invisible?

\begin{theorem}[Characterization of Invisible Dynamics]\label{thm:invisible}
For $X \in \RR_*^{n \times d}$ (i.e., $\rank(X) = d$), a vector field $f$ produces invisible dynamics ($\dot{P} = 0$) if and only if $f(X) = XA$ for some skew-symmetric matrix $A \in \so(d)$.
\end{theorem}

\begin{proof}
($\Leftarrow$) If $f(X) = XA$ with $A^\top = -A$, then:
\begin{equation}
\dot{P} = f(X) X^\top + X f(X)^\top = XAX^\top + XA^\top X^\top = XAX^\top - XAX^\top = 0
\end{equation}

($\Rightarrow$) Suppose $\dot{P} = f(X) X^\top + X f(X)^\top = 0$.
Let $G = X^\top X$, which is invertible since $\rank(X)=d$.
Left-multiply the identity by $(I - XG^{-1}X^\top)$ to project onto $\col(X)^\perp$:
\begin{equation}
0 = (I - XG^{-1}X^\top) f(X) X^\top
\end{equation}
Since $X^\top$ has full row rank, this implies
\begin{equation}
(I - XG^{-1}X^\top) f(X) = 0
\end{equation}
hence $f(X) \in \col(X)$ and there exists $B \in \RR^{d \times d}$ such that $f(X) = XB$.

Substituting back into $\dot{P}=0$ gives
\begin{equation}
0 = X(B + B^\top) X^\top
\end{equation}
and since $X$ has full column rank, this forces $B + B^\top = 0$, i.e.\ $B \in \so(d)$.
Taking $A = B$ completes the proof.
\end{proof}

We can readily interpret the result: invisible dynamics are exactly uniform rotations around the origin in latent space; all other dynamics (attraction, repulsion, non-uniform rotation, drift, \ldots) produce observable changes in network structure.

This implies that the class of invisible dynamics is small (dimension $d(d-1)/2$, the dimension of $\so(d)$), while observable dynamics span a much larger space.

\subsection{Realizable dynamics}\label{sec:realizable}

Beyond gauge freedom, RDPG dynamics face a geometric constraint: the probability matrix $P = XX^\top$ lives on a low-dimensional manifold, so most (symmetric) perturbations $\dot{P}$ are not achievable by dynamics of $X$.

\begin{proposition}[Tangent Space Constraint]\label{prop:tangent}
Let $V \in \RR^{n \times d}$ be an orthonormal basis for $\col(P)$, and $V_\perp \in \RR^{n \times (n-d)}$ span its orthogonal complement.
Any realizable $\dot{P}$ (i.e., $\dot{P} = \dot{X} X^\top + X \dot{X}^\top$ for some $\dot{X}$) satisfies:
\begin{equation}
V_\perp^\top \dot{P} V_\perp = 0
\end{equation}
The realizable tangent space has dimension $nd - d(d-1)/2$.
\end{proposition}

\begin{proof}
Any realizable $\dot{P} = FX^\top + XF^\top$ for some $F \in \RR^{n \times d}$.
Since $\col(X) = \col(V)$, we have $X = VR$ for invertible $R$.
Then:
\begin{equation}
V_\perp^\top \dot{P} V_\perp = V_\perp^\top F R^\top V^\top V_\perp + V_\perp^\top V R F^\top V_\perp = 0
\end{equation}
using $V^\top V_\perp = 0$.
\end{proof}

We can better understand the dynamics by considering a \textbf{block decomposition} of the generating matrices. Observe that any symmetric matrix $M$ decomposes into blocks relative to $(V, V_\perp)$:
\begin{equation}
M = \underbrace{VAV^\top}_{\text{range-range}} + \underbrace{VBV_\perp^\top + V_\perp B^\top V^\top}_{\text{cross terms}} + \underbrace{V_\perp C V_\perp^\top}_{\text{null-null}}
\end{equation}

For realizable $\dot{P}$: the $A$ and $B$ blocks can be arbitrary, but $C = 0$ always.
The null-null block represents directions supported entirely on $\col(P)^\perp$, hence directions that cannot arise from $\dot{P} = FX^\top + XF^\top$ and would generically increase the rank of $P$.
Movements in these directions are forbidden for \emph{fixed} latent dimension $d$ (but they might be achievable for networks that start rank deficient and grow in dimension).

\begin{corollary}[Dimension Count]
For each probability matrix $P$ in the base space $\mathcal{B}$, the set of infinitesimal realizable perturbations $\dot{P}$ forms the tangent space at $P$, and its dimension is:
\begin{equation}
\dim(T_P \mathcal{B}) = nd - d(d-1)/2
\end{equation}
This equals the dimension of $X$-space ($nd$) minus the gauge freedom ($d(d-1)/2$).
\end{corollary}

The block decomposition suggest an immediate diagnostic for the model: if observed dynamics have nonzero structure in the null-null block $V_\perp^\top \dot{P} V_\perp$, this indicates either:
\begin{enumerate}
\item \textbf{Model misspecification}: The true dynamics don't preserve low-rank structure
\item \textbf{Dimensional emergence}: The latent dimension $d$ is increasing, suggesting that new interaction dimensions are emerging
\end{enumerate}

\subsection{The fiber bundle perspective}\label{sec:fiber-bundle}

The gauge freedom in RDPGs has a natural geometric structure that helps us understand what can and cannot be learned from network observations.
The quotient geometry of $\RR_*^{n \times d} / O(d)$ has been developed extensively in the optimization literature, particularly by \citet{absil2008optimization}, \citet{journee2010low}, and \citet{massart2020quotient,massart2019curvature}.
We recall this geometry here, adapting it to the RDPG setting where the connections, curvature, and holonomy structures acquire concrete statistical meaning as obstructions to dynamics estimation.

\begin{remark}
\textbf{Notation (tangent spaces).} For a manifold $\mathcal{M}$ and a point $p \in \mathcal{M}$, we write $T_p \mathcal{M}$ for the tangent space at $p$, and $T \mathcal{M}$ for the tangent bundle (the union of all tangent spaces).
In this paper, at interior points of $\mathcal{E} \subset \RR^{n \times d}$, we identify $T_X \mathcal{E}$ with $\RR^{n \times d}$ (allowable infinitesimal perturbations of $X$). Likewise, $T_P \mathcal{B}$ denotes the space of realizable infinitesimal perturbations $\dot{P}$ at $P$.
\end{remark}

\subsubsection{Why fiber bundles?}

Consider for a moment this scene. While writing this paper, I often worked facing the beach, envying the swimmers out there. Every now and then, a curious shark would move toward the short and be spotted by someone. Then, the local police would go out in a dinghy to look for it. In the (hypothetical) version of that scene we use as intuition here, their radar shows a moving dot: the shark's position projected onto a flat coordinate space centred around the boat. The curious shark moves, very smoothly, in a 3D space: it can go closer or further away from the boat, change it's angle of approach, and diving deeper in the sea. Yet, from the surface trace alone, you cannot tell whether the shark is surfacing or plunging in the depth. Many distinct 3D trajectories can therefore produce the same 2D trace when only this surface projection is observed.

Similarly, in a dynamic RDPG setting we have a latent structure $X$ that evolves in time following a certain regime. Yet, we observe only a function of $X$ which leaks information: we sample graphs from the probability of interaction matrix $P(t) = X(t) X(t)^\top$, and we know $P(t)$ is invariant under $O(d)$ rotations.

In our case, the observable is $P = XX^\top$ (or, in practice, noisy graphs drawn from it), and the latent configuration is $X$ itself.
The latent configuration contains the information we need to talk about velocities and accelerations in latent space, but it cannot be directly measured.
Replacing $X$ by a rotated version $XQ$ does not change $P$. It is a different representative of the same ``surface trace.'' (This is only an analogy. The gauge direction is not literally a physical depth, but the point is the same: there is a hidden direction you cannot see from the base space.)

A \textbf{fiber bundle} formalizes this structure and allows us to study the relationship between the base space and the total (latent) space.
We call the space of all possible observables, e.g.\ $P(t)$, the \textbf{base space} $\mathcal{B}$.
Above each observable $P$, there is a ``fiber'' of equivalent latent configurations $\{XQ : Q \in O(d)\}$, all producing the same $P$.
We call the full space of latent configurations $X$ the \textbf{total space} $\mathcal{E}$. We are endowed with a projection $\pi(X) = XX^\top$ dropping from the total space to the base space.

In a fiber bundle we can consider \textbf{lifts} that take trajectories from the base to the total space: given a path $P(t)$ of evolving observables in the base space, we ``lift'' it to a path $X(t)$ in the total space.
Alas, there are infinitely many possible choices of what $X(t)$ to pick for any $P(t)$, differing by time-varying gauge choices $Q(t)$.
The \emph{connection} and \emph{curvature} of the bundle tell us which lifts are ``natural'' (no spurious rotation) and whether consistent lifting is possible at all.
When it is not, and the bundle is inherently curved, we get \textbf{holonomy}. This is gauge drift that accumulates even along the most careful trajectory.

We now make this precise.

\subsubsection{The probability constraint}

Before defining the bundle, we must address a constraint glossed over earlier: for $P$ to represent connection probabilities, we need $P_{ij} \in [0, 1]$ for all pairs $i, j$.

Not every $X \in \RR^{n \times d}$ satisfies this.
The constraint $(XX^\top)_{ij} = x_i^\top x_j \in [0, 1]$ defines a semi-algebraic subset of $\RR^{n \times d}$.

\begin{definition}[Valid Latent Positions]
The \textbf{valid configuration space} is:
\begin{equation}
\mathcal{E} = \{X \in \RR^{n \times d} : \rank(X) = d, \quad 0 \le x_i^\top x_j \le 1 \text{ for all } i, j\}
\end{equation}

The \textbf{valid probability space} is:
\begin{equation}
\mathcal{B} = \{P \in \RR^{n \times n} : P = XX^\top \text{ for some } X \in \mathcal{E}\}
\end{equation}
\end{definition}

A sufficient condition for $X \in \mathcal{E}$ is that all rows lie in the \textbf{positive orthant of the unit ball}:
\begin{equation}
x_i \in B_+^d = \{x \in \RR^d : x_k \ge 0 \text{ for all } k, \quad \|x\| \le 1\}
\end{equation}
This ensures $x_i^\top x_j \ge 0$ (non-negative entries) and $x_i^\top x_j \le \|x_i\| \|x_j\| \le 1$ (Cauchy--Schwarz).
However, $X \in B_+^d$ is sufficient but not necessary: many valid configurations exist with some entries outside the positive orthant, as long as all pairwise dot products remain in $[0,1]$.

\textbf{Interior and boundary.}
The interior of $\mathcal{E}$ consists of configurations where all constraints are strict: $0 < x_i^\top x_j < 1$.
The boundary includes configurations where some $P_{ij} = 0$ (nodes with orthogonal positions, hence no connection probability) or $P_{ij} = 1$ (nodes with aligned unit-length positions, hence certain connection).

The fiber bundle structure requires a smooth total space, which holds on the interior: there, $\mathcal{E}$ is an open subset of the smooth manifold $\{X \in \RR^{n \times d} : \rank(X) = d\}$, and all the machinery of connections, curvature, and holonomy applies without caveat.

At the boundary, three issues arise.
First, the feasible set $\{X : 0 \le x_i^\top x_j \le 1 \text{ for all } i, j\}$ is defined by polynomial inequalities, and has \emph{corners} where multiple constraints are active simultaneously. Hence, it is a manifold with corners, not a smooth manifold.
Second, at boundary points the tangent space is replaced by a \emph{tangent cone}: not all directions are available, because some would push $P_{ij}$ outside $[0, 1]$.
The horizontal-vertical decomposition of \cref{prop:horizontal} may not respect feasibility: the horizontal component of a feasible velocity can point outside the constraint set.
Third, if a trajectory reaches $P_{ij} = 0$ or $P_{ij} = 1$, the unconstrained ODE $\dot{X} = f(X)$ may attempt to leave the feasible region, requiring constraint handling (projection, reflection, or barrier terms as discussed in \cref{sec:constructive}).

In practice, the boundary issue is rarely binding for temporal networks of interest: edge probabilities that are exactly 0 or 1 correspond to nodes that never interact, which are typically excluded from temporal network analysis (as they are ``undetectable''), or always interact (which mean they don't have stochastic variability in the network structure).
For the remainder of this paper, we work on a connected component of the interior of $\mathcal{E}$, where the geometry is smooth and the bundle structure is clean.
This ensures that trajectories $X(t)$ do not cross rank-deficient strata where the fiber bundle structure degenerates.

\subsubsection{Principal bundle structure}

\begin{definition}[RDPG Principal Bundle]
On the interior of the valid spaces, $(\mathcal{E}, \mathcal{B}, \pi, O(d))$ forms a principal fiber bundle:
\begin{itemize}
\item \textbf{Total space} $\mathcal{E}$: valid latent configurations (interior)
\item \textbf{Base space} $\mathcal{B}$: valid probability matrices (interior)
\item \textbf{Projection} $\pi: \mathcal{E} \to \mathcal{B}$: sends $X \mapsto XX^\top$
\item \textbf{Structure group} $O(d)$: acts on $\mathcal{E}$ by $X \cdot Q = XQ$
\end{itemize}

The \textbf{fiber} over $P$ is $\pi^{-1}(P) = \{XQ : Q \in O(d)\} \simeq O(d)$.
\end{definition}

This is a \textbf{principal bundle} because $O(d)$ acts freely (no $X$ is fixed by any non-identity $Q$) and transitively on each fiber (any two lifts of $P$ differ by some $Q$).

\subsubsection{Decomposing motion: vertical and horizontal}

At each $X \in \mathcal{E}$, we can ask: which directions of motion change $P$, and which don't?

The \textbf{vertical subspace} $\mathcal{V}_X$ consists of directions along the fiber. Motion in these directions changes $X$ but not $P = XX^\top$:
\begin{equation}
\mathcal{V}_X = \ker(d\pi_X) = \{X\Omega : \Omega \in \so(d)\}
\end{equation}
These are exactly the \textbf{invisible dynamics} from \cref{thm:invisible}: infinitesimal rotations $\dot{X} = X\Omega$ with skew-symmetric $\Omega$.

The \textbf{horizontal subspace} $\mathcal{H}_X$ consists of directions transverse to the fiber. Motion in these directions do change $P$ (and $X$):
\begin{equation}
\mathcal{H}_X = \{Z \in T_X \mathcal{E} : X^\top Z \in \Sym(d)\}
\end{equation}

Every tangent vector decomposes uniquely into vertical and horizontal parts:
\begin{equation}
T_X \mathcal{E} = \mathcal{V}_X \oplus \mathcal{H}_X
\end{equation}

\begin{proposition}[Horizontal Characterization]\label{prop:horizontal}
A tangent vector $\dot{X} \in T_X \mathcal{E}$ is horizontal if and only if $X^\top \dot{X}$ is symmetric.
\end{proposition}

\begin{proof}
Write $\dot{X} = X\Omega + H$ with $\Omega \in \so(d)$ (vertical part) and $H \in \mathcal{H}_X$ (horizontal part).
Then $X^\top \dot{X} = X^\top X \Omega + X^\top H$.
The term $X^\top H$ is symmetric by definition of $\mathcal{H}_X$.
The term $(X^\top X) \Omega$ is symmetric only if $\Omega = 0$, since the product of a positive definite symmetric matrix with a nonzero skew-symmetric matrix is never symmetric.
Thus $X^\top \dot{X}$ is symmetric iff $\Omega = 0$ iff $\dot{X}$ is purely horizontal.
\end{proof}

The horizontal condition $X^\top \dot{X} \in \Sym(d)$ is a \textbf{gauge-fixing condition}: it picks out, among all ways to move in $\mathcal{E}$, the ones with no ``wasted motion'' along the fiber.


\subsubsection{The connection 1-form: extracting gauge components}

The choice of horizontal subspaces $\mathcal{H}_X$ varying smoothly over $\mathcal{E}$ is called an \textbf{Ehresmann connection}.
It provides a consistent way to separate ``observable'' from ``gauge'' directions throughout the bundle.
Once we have a connection, we can define parallel transport (moving along the base while staying horizontal) and curvature (measuring how the connection twists).

The Ehresmann connection can be encoded by a \textbf{connection 1-form} $\omega$ that extracts the vertical (gauge) component of any motion:

\begin{definition}[Connection 1-Form]
For each configuration $X \in \mathcal{E}$, the tangent space $T_X \mathcal{E}$ is the vector space of allowable infinitesimal perturbations $Z$ of $X$ (at interior points, it coincides with $\RR^{n \times d}$).
The connection 1-form is the map that assigns to each tangent vector $Z \in T_X \mathcal{E}$ its vertical (gauge) component:
\begin{equation}
\omega_X: T_X \mathcal{E} \to \so(d)
\end{equation}
defined by
\begin{equation}
\omega_X(Z) = \Omega \quad \text{where} \quad (X^\top X) \Omega + \Omega (X^\top X) = X^\top Z - Z^\top X
\end{equation}
That is, $\omega_X(Z)$ is the unique skew-symmetric matrix $\Omega$ solving this Lyapunov equation.
\end{definition}

The horizontal space $\mathcal{H}_X = \{Z : X^\top Z \text{ symmetric}\}$ is not an arbitrary object: it is the \emph{metric connection} induced by the Frobenius inner product.
Equivalently, $\omega_X(Z)$ minimizes $\|Z - X\Omega\|_F^2$ over $\Omega \in \so(d)$: it extracts the gauge component with minimal kinetic energy.
This is the standard Riemannian connection on quotient manifolds~\citep{absil2008optimization,massart2020quotient}, ensuring that horizontal lifts are geodesics when projected appropriately.

To derive its expression, remember that any tangent vector decomposes as $Z = X\Omega + H$ with $\Omega \in \so(d)$ and $X^\top H$ symmetric.
Hence, computing $X^\top Z - Z^\top X$: since $\Omega^\top = -\Omega$, we have $Z^\top X = -\Omega (X^\top X) + H^\top X$.
Thus $X^\top Z - Z^\top X = (X^\top X) \Omega + \Omega (X^\top X) + (X^\top H - H^\top X)$.
Since $X^\top H$ is symmetric, $(X^\top H - H^\top X) = 0$, giving the Lyapunov equation.
Uniqueness holds because the Lyapunov operator $\Omega \mapsto G\Omega + \Omega G$ is invertible when $G = X^\top X$ is positive definite: in the eigenbasis of $G$ (with eigenvalues $\lambda_1, \ldots, \lambda_d$), this operator maps each entry $\Omega_{\iota\gamma}$ to $(\lambda_\iota + \lambda_\gamma) \Omega_{\iota\gamma}$ independently, so it acts as a diagonal operator on the vector space of skew-symmetric matrices, with all eigenvalues $\lambda_\iota + \lambda_\gamma > 0$.
At rank-deficient points ($\rank(X) < d$), $G$ has a zero eigenvalue, the Lyapunov operator acquires a kernel, and the horizontal-vertical decomposition ceases to be unique.
This is the algebraic counterpart of the geometric fact that the fiber bundle structure degenerates at the boundary of $\mathcal{E}$, and foreshadows the role of the spectral gap in \cref{prop:curv-spectral-gap}.

The connection 1-form $\omega_X(\dot{X})$ is the ``instantaneous rotation rate'' of the motion $\dot{X}$: if $\omega_X(\dot{X}) = 0$, the motion is purely horizontal (no gauge component); if $\omega_X(\dot{X}) = \Omega$, then $\dot{X}$ includes a rotation at rate $\Omega$.

The connection satisfies three key properties:
\begin{enumerate}
\item $\ker(\omega_X) = \mathcal{H}_X$: horizontal vectors have zero gauge component
\item $\omega_X(X\Omega) = \Omega$: vertical vectors are identified with their rotation rate
\item Equivariance: $\omega$ transforms appropriately under gauge changes
\end{enumerate}

\subsubsection{Horizontal lifts and parallel transport}

Suppose we observe a trajectory $P(t)$ in the base space (observable dynamics).
We want to ``lift'' this to a trajectory $X(t)$ in the total space.
But there are infinitely many lifts, so which one should we choose?

The \textbf{horizontal lift} is the unique lift with no gauge drift:
\begin{equation}
\pi(X(t)) = P(t) \quad \text{and} \quad \omega_{X(t)}(\dot{X}(t)) = 0
\end{equation}

The second condition says the velocity $\dot{X}(t)$ has no vertical component at any time: the trajectory moves purely horizontally.

\textbf{Existence and uniqueness:}
Given $P(t)$ and an initial lift $X(0)$ with $\pi(X(0)) = P(0)$, the horizontal lift exists and is unique~\citep{kobayashi1963foundations}.
The horizontal condition $\omega = 0$ defines an ODE on $\mathcal{E}$ whose solutions project to $P(t)$.

\textbf{Why horizontal lifts matter:}
Among all trajectories $X(t)$ projecting to $P(t)$, the horizontal lift is the ``gauge-canonical'' one that tracks the observable dynamics without introducing spurious rotation.
If we could compute horizontal lifts from data, we would solve the gauge problem.
The difficulty is that in practice we observe noisy adjacency matrices, not $P(t)$ directly (see \cref{sec:p-dynamics} for a detailed treatment of this estimation problem).

\textbf{Computing horizontal lifts.}
The horizontal lift formula connects directly to the block decomposition of \cref{prop:tangent}.
Let $P = V \Lambda V^\top$ where $V \in \RR^{n \times d}$ has orthonormal columns and $\Lambda = \diag(\lambda_1, \ldots, \lambda_d)$ with $\lambda_\iota > 0$.
A natural lift is $X = V \Lambda^{1/2}$, giving $X^\top X = \Lambda$.

Given realizable $\dot{P}$, the block decomposition yields:
\begin{equation}
\dot{P} = V A V^\top + V B V_\perp^\top + V_\perp B^\top V^\top
\end{equation}
where $A = V^\top \dot{P} V$ (symmetric, $d \times d$) and $B = V^\top \dot{P} V_\perp$ (the cross term).

We seek the horizontal lift $\dot{X} = V S + V_\perp T$ satisfying:
\begin{enumerate}
\item \emph{Projection}: $\dot{P} = \dot{X} X^\top + X \dot{X}^\top$
\item \emph{Horizontality}: $X^\top \dot{X} = \Lambda^{1/2} S$ is symmetric
\end{enumerate}

From the projection condition, matching blocks:
\begin{itemize}
\item Range-range: $A = S \Lambda^{1/2} + \Lambda^{1/2} S^\top$
\item Cross terms: $B = T^\top \Lambda^{1/2}$, giving $T = B^\top \Lambda^{-1/2}$
\end{itemize}

The horizontality condition requires $S = \Lambda^{-1/2} \Sigma$ for some symmetric $\Sigma$.
Substituting into the range-range equation and setting $\tilde{\Sigma} = \Lambda^{-1/2} \Sigma \Lambda^{-1/2}$:
\begin{equation}\label{eq:horizontal-lift-lyapunov}
\Lambda \tilde{\Sigma} + \tilde{\Sigma} \Lambda = A
\end{equation}

A small notational point: $A = V^\top \dot{P} V$ is already the range-range block expressed in the eigenbasis of $P = V \Lambda V^\top$. That is why the right-hand side is simply $A$, with no additional conjugation by powers of $\Lambda$.

This is the same Lyapunov structure as in the connection 1-form above. And not coincidentally: both of them involve separating gauge from observable components. The following formula is essentially the horizontal projection of \citet{absil2008optimization} and \citet{massart2020quotient}, specialized to the spectral decomposition of $P$.

\begin{proposition}[Horizontal Lift Formula]
Given $P = V \Lambda V^\top$ and realizable $\dot{P}$ with blocks $A = V^\top \dot{P} V$ and $B = V^\top \dot{P} V_\perp$, the horizontal lift at $X = V \Lambda^{1/2}$ is:
\begin{equation}
\dot{X} = V \tilde{\Sigma} \Lambda^{1/2} + V_\perp B^\top \Lambda^{-1/2}
\end{equation}
where $\tilde{\Sigma}$ is the symmetric solution to \cref{eq:horizontal-lift-lyapunov}, given elementwise by:
\begin{equation}
\tilde{\Sigma}_{\iota\gamma} = \frac{A_{\iota\gamma}}{\lambda_\iota + \lambda_\gamma}
\end{equation}
\end{proposition}

\begin{proof}
The derivation above (\cref{eq:horizontal-lift-lyapunov} and the preceding block algebra) constructs $\dot{X}$ satisfying both the projection condition $\dot{P} = \dot{X} X^\top + X \dot{X}^\top$ and the horizontality condition $X^\top \dot{X}$ symmetric.
The Lyapunov equation has a unique symmetric solution because $\lambda_\iota + \lambda_\gamma > 0$ for all $\iota, \gamma$.
\end{proof}

The formula reveals two contributions to horizontal motion:
\begin{itemize}
\item The first term $V \tilde{\Sigma} \Lambda^{1/2}$ handles motion within the current column space of $X$, determined by the range-range block $A$ via a Lyapunov equation
\item The second term $V_\perp B^\top \Lambda^{-1/2}$ handles motion expanding into new directions, determined directly by the cross block $B$
\end{itemize}

In practice, computing horizontal lifts from discrete noisy observations requires interpolation and ODE integration.
We discuss in \cref{sec:constructive} how discrete Procrustes alignment approximates horizontal transport.

\subsubsection{Curvature and holonomy}

A subtle issue: even horizontal lifts can accumulate gauge drift over closed loops.

The \textbf{curvature} of the connection measures how horizontal directions fail to commute:

\begin{definition}[Curvature 2-Form]
For horizontal vector fields $H_1, H_2$ on $\mathcal{E}$:
\begin{equation}
\Omega(H_1, H_2) = -\omega([H_1, H_2])
\end{equation}
where $[\cdot, \cdot]$ is the Lie bracket.
\end{definition}

When curvature is nonzero, traveling horizontally in direction $H_1$ then $H_2$ doesn't end up at the same point as $H_2$ then $H_1$. So, there's a vertical (gauge) discrepancy.

This has a striking consequence for closed loops:

\begin{definition}[Holonomy]
Let $\gamma: [0,1] \to \mathcal{B}$ be a closed curve with $\gamma(0) = \gamma(1) = P$.
The horizontal lift starting at $X$ ends at $X Q_\gamma$ for some $Q_\gamma \in O(d)$.
The \textbf{holonomy} of $\gamma$ is this accumulated rotation $Q_\gamma$.
\end{definition}

A classic result provides a precise qualitative characterization of the curvature/holonomy link. Concretely, the Ambrose--Singer theorem~\citep{kobayashi1963foundations} identifies the Lie algebra of the restricted holonomy group (holonomy of contractible loops) as the one generated by curvature values encountered along horizontal transport. (We will later return to concrete holonomy behavior for specific RDPG dynamics families.)

\begin{theorem}[Holonomy Obstruction (Ambrose--Singer)]
If the curvature is nonzero, then the restricted holonomy group is nontrivial. In particular, there exist closed paths in $\mathcal{B}$ (contractible loops) such that no globally consistent gauge exists: any lift satisfies $X(1) = X(0) Q$ for some $Q \neq I$.
\end{theorem}

The result implies that, if the true dynamics $P(t)$ trace a closed loop (periodic network behavior), the underlying $X(t)$ may not close on itself. Instead, it returns rotated by the holonomy.
This is a fundamental obstruction: even perfect local alignment accumulates global gauge drift over cycles.

We can connect the holonomy obstruction to spectral properties of the (observed) graphs: the curvature of $\mathcal{B}$ depends on the eigenvalues of $P = X X^\top$. Since the nonzero eigenvalues of $P$ are exactly those of $X^\top X$ (the Gram matrix of the latent positions), small $\lambda_d$ arises in several common scenarios.

\begin{remark}
\textbf{When the spectral gap is small (examples).} Small $\lambda_d$ occurs, for example, when:
\begin{itemize}
\item \emph{Latent positions cluster in a lower-dimensional subspace:} if nodes' positions are nearly coplanar in $\RR^d$, the columns of $X$ become nearly linearly dependent.
\item \emph{Stochastic block models with weak community structure:} for SBM with $X = Z B^{1/2}$ (membership $Z$, block matrix $B$), the spectral gap of $P$ reflects that of $B$. When communities are hard to distinguish, $\lambda_d$ is small.
\item \emph{Sparse networks:} if all entries of $X$ scale as $\rho_n \to 0$ (sparsity parameter), then $\lambda_d \sim \rho_n^2 \to 0$.
\end{itemize}
\end{remark}

This matters doubly: not only do geometric obstructions worsen as $\lambda_d \to 0$ (the injectivity radius vanishes as $\sqrt{\lambda_d}$, and curvature blows up when $\lambda_{d-1}$ is also small), but ASE convergence rates also deteriorate (they scale as $O(\sqrt{\log n / \lambda_d})$~\citep{cape2019two}).
Configurations with small spectral gap are thus problematic both statistically (harder to estimate) and geometrically (harder to track gauges). We emphasize this expectation-setting point here because it will reappear later: in \cref{sec:dynamics} and \cref{sec:info-theoretic}, the same spectral quantities controlling curvature and injectivity also control estimation error and Fisher information.

\begin{proposition}[Curvature and Spectral Gap]\label{prop:curv-spectral-gap}
The quotient manifold $\mathcal{B} = \RR_*^{n \times d} / O(d)$ with the Procrustes metric has \emph{sectional curvature} $K$ (a scalar measuring the Gaussian curvature of 2-dimensional sections, distinct from the Lie algebra-valued curvature 2-form $\Omega$ above) given by O'Neill's formula for Riemannian submersions~\citep{oneill1966fundamental}: for orthonormal horizontal vectors $\bar{\xi}, \bar{\eta} \in \mathcal{H}_X$,
\begin{equation}
K(\xi, \eta) = 3 \|\mathcal{A}_{\bar{\xi}} \bar{\eta}\|^2 = \frac{3}{4} \|[\bar{\xi}, \bar{\eta}]^{\mathcal{V}}\|^2
\end{equation}
where $\mathcal{A}_{\bar{\xi}} \bar{\eta} = \frac{1}{2} [\bar{\xi}, \bar{\eta}]^{\mathcal{V}}$ is the O'Neill $A$-tensor and $[\bar{\xi}, \bar{\eta}]^{\mathcal{V}}$ denotes the vertical component of the Lie bracket. Since the total space $\RR_*^{n \times d}$ is flat, all base curvature arises from the $A$-tensor.

\citet{massart2019curvature} compute the sectional curvature explicitly for this quotient:
\begin{itemize}
\item If $d = 1$, the sectional curvature is identically zero (the base space is flat).
\item If $d \ge 2$, the minimum sectional curvature is zero, achieved when the horizontal fields commute.
\item The maximum sectional curvature at $[X]$ diverges when the two smallest eigenvalues $\lambda_{d-1}, \lambda_d$ of $P = XX^\top$ approach zero simultaneously.
\end{itemize}
\end{proposition}

Two related but distinct geometric obstructions emerge from the spectral gap.
The \emph{curvature} depends on the $A$-tensor, which involves the connection coefficients $1/(\lambda_\iota + \lambda_\gamma)$ for $\iota \neq \gamma$.
Since the connection is skew-symmetric, the worst case is $1/(\lambda_{d-1} + \lambda_d)$: curvature blows up only when \emph{both} of the two smallest eigenvalues approach zero.
If $\lambda_d \to 0$ but $\lambda_{d-1}$ remains bounded, the curvature stays finite (the manifold looks locally like a cylinder in the collapsing direction).
By contrast, the \emph{injectivity radius} at $[X]$ equals $\sqrt{\lambda_d}$~\citep{massart2020quotient} and vanishes whenever $\lambda_d \to 0$ alone: beyond this radius, geodesics cease to be length-minimizing and multiple paths can connect the same points.

For dynamics estimation, both phenomena matter.
Large curvature means that even short paths produce substantial holonomy, making local Procrustes alignment inconsistent over moderate distances.
Small injectivity radius means that interpolation between time steps is non-unique, even without curvature.
The $d = 1$ case is degenerate in a useful way: $O(1) = \{\pm 1\}$ is discrete, so the only gauge ambiguity is a global sign flip, and the curvature vanishes because there is no continuous rotation to accumulate.

\subsubsection{Riemannian structure}

The quotient $\mathcal{B} \simeq \mathcal{E} / O(d)$ inherits a natural metric:

\begin{definition}[Quotient Metric]
The Riemannian distance between $[X], [Y] \in \mathcal{B}$ is the \textbf{Procrustes distance}:
\begin{equation}
d_{\mathcal{B}}([X], [Y]) = \min_{Q \in O(d)} \|X - YQ\|_F
\end{equation}
\end{definition}

This connects the geometry to computation: Procrustes alignment computes geodesic distance on the base space~\citep{massart2020quotient}.
The global injectivity radius of $\mathcal{B}$ is zero (since $\lambda_d$ can be arbitrarily small), reflecting the nontrivial topology of the quotient.


\section{Recovering trajectories from spectral embeddings}\label{sec:trajectory-problem}

The obstructions above concern what dynamics are \emph{theoretically} possible.
We now turn to the \emph{practical} problem: even when dynamics are observable and realizable, recovering trajectories from data is hard because spectral embedding introduces arbitrary gauge transformations at each time step.

\subsection{ASE introduces arbitrary gauge transformations}

At each time $t$, ASE computes the eigendecomposition of $A^{(t)}$.
The eigenvectors are determined only up to sign (and rotation within repeated eigenspaces).
This means:
\begin{equation}
\hat{X}^{(t)} = X^{(t)} R^{(t)} + E^{(t)}
\end{equation}
where $E^{(t)}$ is statistical noise and $R^{(t)} \in O(d)$ is the gauge factor.
Although $R^{(t)}$ is deterministic (fixed by the eigensolver's conventions, e.g.\ sign of the first nonzero entry), it is \emph{discontinuous} in $t$: standard eigensolvers enforce conventions that jump whenever an entry crosses zero or eigenvalues change ordering.
These jumps are uncorrelated with the dynamics, so $R^{(t)}$ behaves effectively like an arbitrary selection from the fiber $O(d)$ at each time step.

\textbf{The key point}: Even if the true positions $X^{(t)}$ evolve smoothly, the estimates $\hat{X}^{(t)}$ jump erratically because the $R^{(t)}$ are unrelated across time.
In principle, a smooth choice of eigenvectors exists when eigenvalues are simple and vary analytically, but standard numerical eigensolvers do not enforce this choice.

\subsection{Finite differences fail}

Consider estimating velocity via finite differences:
\begin{equation}
\hat{\dot{X}}^{(t)} = \frac{\hat{X}^{(t + \delta t)} - \hat{X}^{(t)}}{\delta t}
\end{equation}

Substituting the gauge-contaminated estimates:
\begin{equation}
\hat{\dot{X}}^{(t)} = \frac{X^{(t+\delta t)} R^{(t+\delta t)} - X^{(t)} R^{(t)}}{\delta t} + O(E)
\end{equation}

Even ignoring noise $E$, this is dominated by the gauge jump $R^{(t+\delta t)} - R^{(t)}$, which is $O(1)$ regardless of $\delta t$.
As the ``velocity'' diverges as $\delta t \to 0$, we're measuring gauge jumps, not dynamics.

\subsection{Pairwise alignment and error accumulation}\label{sec:alignment-accumulation}

One might try aligning consecutive embeddings via Procrustes:
\begin{equation}
Q^{(t)} = \arg \min_{Q \in O(d)} \|\hat{X}^{(t+1)} - \hat{X}^{(t)} Q\|_F
\end{equation}

This has a closed-form solution via SVD and finds the best rotation to match adjacent frames.
However, sequential pairwise alignment suffers from fundamental limitations:

\begin{enumerate}
\item \emph{Local, not global:} Each alignment minimizes error between adjacent frames but doesn't ensure consistency across the full trajectory.

\item \emph{Error accumulation:} The ASE gauge $R^{(t)}$ behaves like an arbitrary $O(1)$ rotation at each time step, unrelated across time. Pairwise Procrustes targets the relative gauge $(R^{(t)})^{-1} R^{(t+1)}$ between consecutive frames, with error controlled by the ASE noise $E^{(t)}$. When these noise-induced errors are roughly independent across steps, the accumulated rotation error after aligning $T$ frames behaves diffusively, scaling like $O(\sqrt{T} \sigma)$ where $\sigma$ is the per-step noise level.

\item \emph{Holonomy is a separate obstruction:} The $O(\sqrt{T} \sigma)$ diffusive error is the \emph{statistical} difficulty of alignment. The holonomy obstruction (\cref{sec:fiber-bundle}) is \emph{geometric/topological} and exists even with perfect data ($\sigma = 0$): for dynamics with nontrivial holonomy, no globally consistent gauge exists over closed loops, regardless of how accurately each pairwise alignment is performed. The two obstructions interact but should not be conflated.

\item \emph{No dynamical constraint:} Even a perfectly aligned sequence need not correspond to a dynamically meaningful lift: Procrustes alignment is purely geometric and does not enforce that $\dot{X} = f(X)$ for any structured family.
\end{enumerate}

\textbf{Noise in spectral embeddings:}
RDPG spectral embeddings have estimation error $\|\hat{X} - XQ\|_F = O_p(\sqrt{n})$ for some $Q \in O(d)$, giving per-node error $O_p(1/\sqrt{n})$~\citep{athreya2017statistical,cape2019two}.
For pairwise Procrustes between consecutive frames, this translates to rotation estimation error that depends on both $n$ and the separation between frames.

When the true trajectory moves slowly (small $\|X^{(t+1)} - X^{(t)}\|$), the signal-to-noise ratio for alignment degrades. In these cases, we're trying to detect small true rotations against a background of estimation noise.

\subsection{Why existing joint embedding methods don't help}\label{sec:why-not-uase}

Joint embedding methods like UASE~\citep{gallagher2021spectral} embed all time points simultaneously, producing gauge-consistent estimates.
However, they assume a different generative model.

\textbf{UASE model (Multilayer RDPG):}
\begin{equation}
P^{(t)}_{ij} = x_i \Lambda^{(t)} y_j^{(t)\top}
\end{equation}
where $x_i$ is a \emph{fixed} ``identity'' and $y_j^{(t)}$ is a time-varying ``activity.''

\textbf{Our model (temporal RDPG with ODE dynamics):}
\begin{equation}
P^{(t)}_{ij} = x_i^{(t)\top} x_j^{(t)}
\end{equation}
where the \emph{same} position $x_i^{(t)}$ evolves according to $\dot{x}_i = f(x_i, \ldots)$.

We can express the mismatch through the shared-subspace assumption.
UASE assumes that the concatenated matrix $[A^{(1)}; A^{(2)}; \ldots; A^{(T)}]$ has a \emph{common invariant column space} of dimension $d$: all time slices share the same left-singular vectors.
Under ODE dynamics on latent positions, the column space of $P^{(t)} = X^{(t)} X^{(t)\top}$ can evolve over time; for non-commuting families (notably Laplacian dynamics), eigenvectors rotate and a single invariant subspace is typically absent.
UASE therefore fits a static average subspace to what is inherently an evolving one, compressing the dynamical signal into a deformation of ``activity scores'' within a fixed basis.
In short: joint embeddings enforce a shared subspace; ODE-driven temporal RDPGs need not have one (except in special commuting/symmetric regimes).
Applying UASE to data from our model distorts the recovered trajectory, attributing genuine positional evolution to time-varying activities against fixed identities.

Similar issues affect Omnibus embedding~\citep{levin2017central} and COSIE~\citep{arroyo2021inference}: all assume a shared subspace structure that ODE dynamics do not preserve.

\subsection{Why Bayesian smoothing approaches are insufficient}\label{sec:bayesian-smoothing}

Recent work~\citep{loyal2025} proposes Bayesian inference for dynamic RDPGs using hierarchical priors on latent positions.
By placing priors on successive differences (or equivalently, using Gaussian processes with smooth kernels), the resulting trajectories are smooth: velocities and accelerations are well-defined and continuous.

In this section, we make one conceptual point: smoothness is necessary but not sufficient for dynamical consistency. We then formalize this gap (via a measure-zero result) and conclude with practical diagnostics.

However, smoothness is necessary but not sufficient for dynamical consistency.
An ODE $\dot{X} = f(X)$ constrains the velocity to be a \emph{function of the current state}; a smoothness prior constrains velocities to vary continuously, but does not require that $\dot{X}(t)$ is determined by $X(t)$.

To be clear about scope: the critique in this section is aimed at \emph{smoothing-only} models that place priors directly on trajectories (or their increments) without including $f$ and the ODE constraint in the generative model. Bayesian inference with an explicit ODE/SDE structure is a different object entirely (see the comparison table below).

\begin{definition}[Dynamical Consistency]
A trajectory $X(t)$ is \textbf{dynamically consistent} with respect to a function class $\mathcal{F}$ if there exists $f \in \mathcal{F}$ such that $\dot{X}(t) = f(X(t))$ for all $t$.
\end{definition}

Hierarchical Bayesian priors ensure that $X(t)$, $\dot{X}(t)$, and higher derivatives are all smooth, but they do not enforce the state-dependence constraint $\dot{X}(t) = f(X(t))$ for any $f$.
This distinction is the gap between kinematic regularity and dynamical structure.

\textbf{Prior support and the ODE solution manifold.}
For a given initial condition $X(0) = X_0$ and dynamics $\dot{X} = f(X)$, the solution $X(t)$ traces a \emph{unique} curve: there is no uncertainty in the trajectory given $(f, X_0)$.
The space of all ODE solutions (over all $f \in \mathcal{F}$ and $X_0$) forms a \emph{finite-dimensional manifold} in the infinite-dimensional space of smooth paths.
A hierarchical smoothness prior assigns positive probability to all smooth paths, a much larger set than just those satisfying some ODE. So, the ODE solution manifold has measure zero under such priors.

\begin{proposition}[Measure Zero for Finite-Dimensional ODE Families]\label{prop:measure-zero}
Let $\mathcal{F} = \{f_\theta : \theta \in \Theta \subset \RR^p\}$ be a finite-dimensional $C^1$ family of vector fields on $\RR^{n \times d}$, and assume each ODE $\dot{X} = f_\theta(X)$ has a unique solution on $[0,T]$ for every initial condition.
Let $\mu$ be a non-degenerate Gaussian measure on $C^k([0,T], \RR^{n \times d})$ with full support.
Define
\begin{equation}
\mathcal{M}_\mathcal{F} = \{X(\cdot) : \dot{X} = f_\theta(X) \text{ for some } \theta \in \Theta\}.
\end{equation}
Then $\mu(\mathcal{M}_\mathcal{F}) = 0$.
\end{proposition}

\begin{proof}
See \cref{app:deferred-proofs}.
\end{proof}

\begin{proposition}[Zero Posterior Mass Without an ODE-Constrained Prior]
Let $\Pi_0 = \mu$ be a prior on path space with $\mu(\mathcal{M}_\mathcal{F}) = 0$.
For observed data $y$, define a posterior via
\begin{equation}
\Pi(A \mid y) = \frac{\int_A L(y \mid X) \, d\mu(X)}{\int L(y \mid X) \, d\mu(X)}
\end{equation}
with measurable likelihood $L$ and finite, positive normalizer.
Then $\Pi(\mathcal{M}_\mathcal{F} \mid y) = 0$ for every such dataset $y$.
\end{proposition}

\begin{proof}
The posterior is absolutely continuous with respect to $\mu$, with Radon--Nikodym derivative proportional to $L(y \mid X)$.
Therefore every $\mu$-null set is posterior-null, in particular $\mathcal{M}_\mathcal{F}$.
\end{proof}

Thus, unless the ODE constraint is built into the prior/model, Bayesian updating cannot assign positive posterior mass to exactly dynamically consistent trajectories.
Asymptotic concentration near (rather than on) $\mathcal{M}_\mathcal{F}$ depends on model misspecification and is a separate question.
This is the distinction between \emph{interpolation} and \emph{dynamics learning}: both produce smooth curves through the data, but only the latter respects the constraint that velocity is state-dependent.

\textbf{Diagnostics for detecting dynamical inconsistency.}
Several tests can distinguish smooth-but-dynamically-inconsistent trajectories from genuine ODE solutions.
First, a \emph{state-dependence test}: for autonomous dynamics $\dot{X} = f(X)$, if the trajectory passes through similar positions at two different times ($X(t_1) \approx X(t_2)$), then the inferred velocities must also be similar ($\dot{X}(t_1) \approx \dot{X}(t_2)$); smoothing priors do not enforce this, and large velocity discrepancies at revisited states indicate dynamical inconsistency.
Second, an \emph{ODE residual test}: fit candidate dynamics $\hat{f}$ to the smoothed trajectory and compute residuals $r(t) = \dot{\hat{X}}(t) - \hat{f}(\hat{X}(t))$; for dynamically consistent trajectories $\|r(t)\|$ should be small, while for smooth-but-wrong trajectories residuals will be systematically large.
Third, \emph{flow consistency}: integrate the fitted $\hat{f}$ forward from various initial conditions along the trajectory; dynamically consistent trajectories will track the integrated flow, while interpolated trajectories will diverge.
Finally, \emph{medium-to-long-term forecasting} discriminates the two approaches: kinematic extrapolation assumes velocity and acceleration persist, while dynamic extrapolation adapts velocity to the current state via $f(X)$, and only the latter remains accurate when the trajectory enters new regions of state space.

\textbf{Alternative approaches with dynamical consistency.}

\begin{table}[ht]
\centering
\begin{tabular}{lll}
\toprule
\textbf{Approach} & \textbf{Smoothness} & \textbf{Dynamical Consistency} \\
\midrule
Hierarchical GP prior & Yes ($C^k$) & No: interpolation only \\
SDE $dX = f(X)\,dt + \sigma\,dW$ & No (H\"older $< 1/2$) & Yes, as $\sigma \to 0$ \\
Neural ODE & Yes ($C^k$) & Yes (by construction) \\
GP-ODE (NPODE)~\citep{heinonen2018learning} & Yes ($C^1$) & Yes: learns $f$ as GP \\
Structure-constrained (ours) & Yes & Yes: family $\mathcal{F}$ enforced \\
\bottomrule
\end{tabular}
\caption{Comparison of approaches: smoothness vs dynamical consistency.}
\end{table}

The SDE formulation $dX = f(X)\,dt + \sigma\,dW$ provides a principled bridge: Freidlin--Wentzell theory~\citep{freidlin1998random} shows that as $\sigma \to 0$, solutions concentrate around ODE solutions with rate function $J_T(\phi) = \frac{1}{2} \int_0^T \|\dot{\phi}(t) - f(\phi(t))\|^2 \, dt$.
This rate function is precisely the \emph{dynamical consistency penalty}, measuring how far a path deviates from being an ODE solution.

\begin{remark}
\textbf{Limitation (trajectory recovery).} Aligning spectral embeddings to recover continuous-time trajectories remains a difficult open problem.
The approaches discussed above address related tasks, but they do not resolve trajectory recovery for ODE-driven latent positions in RDPGs in a provable way.
To our knowledge, there is no existing method that provably recovers trajectories from ODE dynamics on RDPG latent positions.

In particular, sequential alignment suffers from error accumulation over long trajectories (\cref{sec:alignment-accumulation}).
Moreover, smoothness alone does not guarantee dynamical consistency, so interpolation can be mistaken for genuine dynamics learning (\cref{sec:bayesian-smoothing}).
Finally, holonomy implies that even in the absence of statistical noise, globally consistent gauge choices may fail to exist over loops in the base space (\cref{sec:fiber-bundle}).

These considerations motivate our focus on dynamics structure as a source of additional constraints for alignment; we return to this constructive question in \cref{sec:constructive} after analyzing concrete dynamics families in \cref{sec:dynamics}.
\end{remark}



\section{Dynamics on RDPGs}\label{sec:dynamics}

The geometric framework of \cref{sec:obstructions} characterizes the abstract structure of gauge freedom, curvature, and holonomy, and \cref{sec:trajectory-problem} demonstrates the practical difficulties of recovering trajectories from spectral embeddings.
We analyze concrete dynamics families, their observable consequences, and the interplay between geometric and statistical obstructions to estimation.

\subsection{Families of RDPG dynamics}\label{sec:dynamics-families}

We catalog concrete families of dynamics on RDPG latent positions.
These families will later serve as inductive bias for the alignment problem (\cref{sec:constructive}); the holonomy and statistical analyses that follow depend on their specific structure.

\textbf{Linear dynamics.}
The simplest family is $\dot{X} = NX$ with $N \in \RR^{n \times n}$, where $N_{ij}$ determines how node $j$'s position affects node $i$'s velocity.
By \cref{prop:horizontal}, this is horizontal if and only if $X^\top NX$ is symmetric; a sufficient condition is $N = N^\top$, which gives $X^\top NX$ symmetric for all $X$.
The induced $P$-dynamics are $\dot{P} = NP + PN^\top$, which reduces to the Lyapunov equation $\dot{P} = NP + PN$ when $N$ is symmetric.
Symmetric linear dynamics have $n(n+1)/2$ free parameters.

\textbf{Polynomial dynamics in $P$.}
A far more parsimonious family replaces the constant $N$ with a polynomial in the probability matrix:
\begin{equation}
\dot{X} = N(P) X, \quad N(P) = \sum_{k=0}^{K} \alpha_k P^k
\end{equation}
Since $P$ is symmetric, $N(P)$ is automatically symmetric and these dynamics are \emph{always horizontal}, regardless of the coefficients $\alpha_k$.
The terms have a clean interpretation: $P^0 = I$ gives self-dynamics (decay or growth), $P^1 = P$ encodes direct neighbor influence, and $P^k$ captures $k$-hop effects through weighted path counts $(P^k)_{ij}$.
Crucially, $P = XX^\top$ is gauge-invariant, so $N(P)$ can be computed from observations without solving the alignment problem.
The family has only $K + 1$ parameters independent of $n$, a dramatic reduction from $n(n+1)/2$ for general symmetric $N$ that is key for identifiability.

\textbf{Graph Laplacian dynamics.}
Setting $\dot{X} = -LX$ with $L = D - P$ and $D = \diag(P\mathbf{1})$ gives diffusion on the graph: entry-wise, $\dot{x}_i = \sum_j P_{ij}(x_j - x_i)$, so each node moves toward the weighted average of its neighbors.
Since $L$ is symmetric, this is horizontal.
Unlike polynomial dynamics, the Laplacian is \emph{not} a polynomial in $P$: the degree matrix $D$ depends on row sums, which mix eigenvector information that polynomials in $P$ cannot access.
This distinction will prove crucial for holonomy (\cref{sec:holonomy-dynamics}).

\textbf{Message-passing dynamics.}
More generally, $\dot{x}_i = \sum_j P_{ij} g(x_i, x_j)$ for some interaction function $g: \RR^d \times \RR^d \to \RR^d$.
Special cases include neighbor attraction ($g(x_i, x_j) = x_j$, giving $\dot{X} = PX$) and Laplacian diffusion ($g(x_i, x_j) = x_j - x_i$).
Horizontality depends on $g$; the locality constraint (velocity depends only on neighbors weighted by $P$) is itself a useful structural assumption.

\textbf{Observable non-horizontal dynamics.}
Not all observable dynamics are horizontal.
Centroid circulation $\dot{x}_i = (x_i - \bar{x}) A$ with $A \in \so(d)$ decomposes as $\dot{x}_i = x_i A - \bar{x} A$: the first term is pure gauge (\cref{thm:invisible}), while the shared drift $-\bar{x} A$ generally shifts dot products, so $\dot{P}$ is typically nonzero except in special symmetric configurations.
However, $X^\top \dot{X} = CA$ where $C$ is the sample covariance of positions; this is symmetric only under additional algebraic constraints (e.g., commutation-compatible structure), so the dynamics are typically non-horizontal and the trajectory spirals through the fibers.
Similarly, differential rotation $\dot{x}_i = x_i A_i$ with node-specific $A_i \in \so(d)$ gives $\dot{P}_{ij} = x_i (A_i - A_j) x_j^\top$, which is observable whenever rotation rates differ.

\textbf{Summary.}

\begin{table}[ht]
\centering
\begin{tabular}{llcl}
\toprule
\textbf{Family} & \textbf{Form} & \textbf{Horizontal?} & \textbf{Parameters} \\
\midrule
Linear symmetric & $\dot{X} = NX$, $N = N^\top$ & Always & $n(n+1)/2$ \\
Polynomial & $\dot{X} = (\sum_k \alpha_k P^k) X$ & Always & $K + 1$ \\
Laplacian & $\dot{X} = -LX$ & Always & 0 (fixed) \\
Message-passing & $\dot{x}_i = \sum_j P_{ij} g(x_i, x_j)$ & If $g$ symmetric & $|\theta|$ \\
Centroid circulation & $\dot{x}_i = (x_i - \bar{x}) A$ & No & $d(d-1)/2$ \\
General linear & $\dot{X} = NX$ & If $X^\top NX$ sym. & $n^2$ \\
\bottomrule
\end{tabular}
\caption{Families of RDPG dynamics. The polynomial family is always horizontal, has few parameters ($K + 1$ independent of $n$), and is gauge-invariant ($N(P)$ computable without alignment). Centroid circulation is observable but not horizontal.}
\label{tab:dynamics-families}
\end{table}

\subsection{Observable dynamics and the Lyapunov equation}\label{sec:p-dynamics}

\Cref{thm:invisible} characterizes what is \emph{lost} to gauge freedom.
We characterize what is \emph{preserved}: the gauge-invariant footprint of latent-space dynamics, i.e.\ what survives after projection to the base space $\mathcal{B}$ of probability matrices.

The key observation uses two ingredients together.
First, the RDPG model gives the factorization $P = XX^\top$, so $P$ inherits smoothness from $X$ and the product rule applies.
Second, a dynamics model (start with the clean case $\dot{X} = NX$) specifies how positions evolve.
Together, these imply a closed equation for $\dot{P}$:
\begin{equation}\label{eq:p-dynamics}
\dot{P} = \dot{X} X^\top + X \dot{X}^\top = NXX^\top + XX^\top N^\top = NP + PN^\top
\end{equation}
When $N$ is symmetric, this becomes the Lyapunov equation $\dot{P} = NP + PN$.

This is \emph{not} a generic statement about temporal networks.
Without the low-rank factorization, $P(t)$ is just an arbitrary time-varying probability matrix whose entries could evolve independently, follow a different matrix ODE, or fail to be differentiable at all.
\Cref{eq:p-dynamics} is specific to the RDPG-as-dynamical-system viewpoint: it is what the latent ODE looks like once you throw away gauge.

The Lyapunov operator $\mathcal{L}_P: \Sym(n) \to \Sym(n)$ defined by $\mathcal{L}_P(N) = NP + PN$ is a linear map on symmetric matrices.
To analyze it, we pass to the eigenbasis of $P$.
Let $P = U \Lambda U^\top$ be the eigendecomposition, and write $\tilde{N} = U^\top N U$ and $\tilde{\dot{P}} = U^\top \dot{P} U$ for the representations in this basis, with indices $\iota, \gamma \in \{1, \ldots, n\}$ labeling eigenvector directions (not nodes).
In this basis, $\mathcal{L}_P$ acts diagonally: $(\mathcal{L}_P(N))_{\iota\gamma} = (\lambda_\iota + \lambda_\gamma) \tilde{N}_{\iota\gamma}$.

For the RDPG setting, where $P = XX^\top$ with $X \in \RR^{n \times d}$, the matrix $P$ has rank $d$ with eigenvalues $\lambda_1 \ge \cdots \ge \lambda_d > 0$ and $\lambda_{d+1} = \cdots = \lambda_n = 0$.
The Lyapunov operator is therefore \emph{not invertible} on all of $\Sym(n)$: when both $\iota, \gamma > d$, the equation $(\lambda_\iota + \lambda_\gamma) \tilde{N}_{\iota\gamma} = \tilde{\dot{P}}_{\iota\gamma}$ reduces to $0 \cdot \tilde{N}_{\iota\gamma} = \tilde{\dot{P}}_{\iota\gamma}$.
The realizability constraint (\cref{prop:tangent}) forces $\tilde{\dot{P}}_{\iota\gamma} = 0$ for this block, so the equation says nothing about $\tilde{N}_{\iota\gamma}$.
For $n = 100$ and $d = 3$, this leaves $(n - d)(n - d + 1)/2 = 4753$ out of $n(n+1)/2 = 5050$ entries of $\tilde{N}$ completely unconstrained by the data.

However, the resolution is clean: the undetermined part of $N$ is exactly the part that does not affect $\dot{X} = NX$ for the given $X$.

\begin{proposition}[Partial Lyapunov Identifiability]\label{prop:lyapunov-invert}
Let $P = XX^\top$ with $\rank(P) = d < n$, and let $V \in \RR^{n \times d}$ and $V_\perp \in \RR^{n \times (n-d)}$ span the range and null space of $P$ respectively.
In the eigenbasis of $P$ (indices $\iota, \gamma$), the Lyapunov equation $\dot{P} = NP + PN$ uniquely determines the following blocks of $\tilde{N} = U^\top N U$:

\begin{itemize}
\item \emph{Range-range block} ($\iota, \gamma \le d$): $\tilde{N}_{\iota\gamma} = \tilde{\dot{P}}_{\iota\gamma} / (\lambda_\iota + \lambda_\gamma)$, giving $d(d+1)/2$ entries.
\item \emph{Cross block} ($\iota \le d, \gamma > d$ or vice versa): $\tilde{N}_{\iota\gamma} = \tilde{\dot{P}}_{\iota\gamma} / \lambda_\iota$, giving $d(n - d)$ entries.
\end{itemize}

The \emph{null-null block} ($\iota, \gamma > d$): $(n-d)(n-d+1)/2$ entries remain unconstrained.

The total number of determined entries is $nd - d(d-1)/2$, matching the dimension of the realizable tangent space.
\end{proposition}

\begin{proof}
In the eigenbasis, the Lyapunov equation gives $(\lambda_\iota + \lambda_\gamma) \tilde{N}_{\iota\gamma} = \tilde{\dot{P}}_{\iota\gamma}$ entry-wise.
For $\iota, \gamma \le d$: $\lambda_\iota + \lambda_\gamma > 0$, so $\tilde{N}_{\iota\gamma}$ is uniquely determined.
For $\iota \le d, \gamma > d$: $\lambda_\iota + \lambda_\gamma = \lambda_\iota > 0$, so $\tilde{N}_{\iota\gamma}$ is uniquely determined.
For $\iota, \gamma > d$: $\lambda_\iota + \lambda_\gamma = 0$, and realizability forces $\tilde{\dot{P}}_{\iota\gamma} = 0$, leaving $\tilde{N}_{\iota\gamma}$ free.
Counting: $d(d+1)/2 + d(n-d) = nd - d(d-1)/2$.
\end{proof}

The undetermined null-null block governs how $N$ acts on vectors orthogonal to $\col(X)$.
Since $X$ has no component in this subspace, $NX$ does not depend on this block: the dynamics $\dot{X} = NX$ are insensitive to it.
The Lyapunov equation determines $N$ on exactly the subspace that matters for the dynamics, and leaves unconstrained exactly the subspace that is dynamically irrelevant.

In summary, symmetric linear dynamics are identifiable from $(P, \dot{P})$ up to the dynamically irrelevant null-null block, without ever choosing a gauge.
This perspective extends naturally to polynomial dynamics: if $\dot{X} = N(P) X$ with $N(P) = \sum_k \alpha_k P^k$, then
\begin{equation}
\dot{P} = N(P) P + P N(P)
\end{equation}
and the coefficients $\alpha_k$ can, in principle, be recovered by fitting $\dot{P}$ in the span of the symmetric matrices $\{P^k P + P P^k\}_{k=0}^K$.
Here the polynomial structure is doing real work: it pins down the otherwise-underdetermined action of $N$ off $\col(P)$, eliminating the free null-null block by construction.

However, two fundamental obstacles prevent direct application.

\textbf{We do not observe $P(t)$.}
In practice, we observe adjacency matrices $A^{(t)}$ with $A_{ij}^{(t)} \sim \text{Bernoulli}(P_{ij}(t))$.
Entry-wise, $A_{ij}^{(t)}$ is an unbiased estimate of $P_{ij}(t)$, with variance $P_{ij}(1 - P_{ij})$.
If we have $m$ independent graphs per time point, averaging $\bar{A}^{(t)} = \frac{1}{m} \sum_{\ell=1}^m A_\ell^{(t)}$ reduces this per-entry variance by a factor $m$.

But Lyapunov inversion (and most geometry) depends on the \emph{spectrum} of $P$, not just its entries, and spectral accuracy depends on both $n$ and $m$.
Because eigendecomposition is nonlinear, the eigenvalues/eigenvectors of $\bar{A}$ are biased at finite sample size.
Asymptotically, the leading perturbation has conditional mean zero and the residual bias is small~\citep{athreya2016limit,cape2019two}, but the inverse map in \cref{prop:lyapunov-invert} divides by $(\lambda_\iota + \lambda_\gamma)$.
That division amplifies exactly the directions that are already ill-conditioned when the spectral gap is small. This is the same $1/\lambda_d$ sensitivity that showed up geometrically in \cref{sec:fiber-bundle}, now wearing an algebraic hat.

\textbf{The parameter count of $N$ on the relevant subspace.}
Even after throwing away the dynamically irrelevant null-null block, the remaining degrees of freedom can still be large: the range-range and cross blocks together have $nd - d(d-1)/2$ entries.

If $N$ is constant and you could somehow access a well-estimated pair $(P, \dot{P})$, the constraint count matches the unknown count.
But as soon as $N$ varies with time/state, you need enough \emph{excitation} in the trajectory to separate parameters.
Even in polynomial dynamics, where the coefficients $\alpha_k$ are constant, identifying them requires that $P(t)$ moves through enough of the base space to distinguish the different powers $P^k$.
This is why the polynomial family is attractive: $K+1$ parameters (independent of $n$) is not just parsimonious, it is structurally compatible with what the data can hope to constrain.
Richer families quickly become underdetermined without strong structural assumptions or regularization.

\begin{remark}
The $P$-dynamics perspective and the fiber bundle perspective address different aspects of the same problem.
\Cref{prop:lyapunov-invert} says that the dynamics are identifiable \emph{in principle} from the gauge-invariant observable $P$, up to the dynamically irrelevant null-null block.
The fiber bundle theory (\cref{sec:fiber-bundle}) addresses the harder question: what happens when we \emph{must} work with $X$ (e.g., because we need latent positions for interpretation or for the UDE pipeline), and how the geometry of the quotient space governs the difficulty of doing so.
The two perspectives are complementary: the algebraic result tells us what information is available; the geometric framework tells us what it costs to extract it.
\end{remark}

The Euclidean mirror of \citet{athreya2024euclidean} provides a concrete, estimable gauge-invariant trajectory summary. Their inter-time distance
\begin{equation}
\hat{d}_{MV}(\hat{X}_t, \hat{X}_{t'}) = \min_{W \in O(d)} n^{-1/2} \|\hat{X}_t - \hat{X}_{t'} W\|_2
\end{equation}
is a Procrustes distance on spectral embeddings. Like our Riemannian distance $d_{\mathcal{B}}([X],[Y]) = \min_{Q \in O(d)} \|X - YQ\|_F$, it quotients out the $O(d)$ gauge and yields a metric on the base space $\mathcal{B}$. The two differ in their choice of norm: $d_{\mathcal{B}}$ uses the Frobenius norm (yielding the Riemannian submersion metric from which our connection and curvature derive), while $d_{MV}$ uses the spectral norm (measuring the worst-case directional variation, after the $1/\sqrt{n}$ scaling). For $d = 1$ they coincide; for $d \ge 2$ the minimizing rotations and the resulting distances are generally distinct.

Despite this difference, the conceptual point is shared: both project from the total space to $\mathcal{B}$ and thereby discard gauge.
\citet{athreya2024euclidean} prove that $\hat{d}_{MV}$ consistently estimates $d_{MV}$ from adjacency spectral embeddings at rate $O(\log(n) / \sqrt{n})$ (their Theorem~6), and then prove consistency of the CMDS mirror estimate from the estimated distance matrix (their Theorem~7), providing a concrete estimation guarantee for a gauge-invariant quantity. From the dynamics-learning perspective, however, any scalar inter-time distance (whether Frobenius- or spectral-norm-based) captures \emph{how far} $P$ has moved but not \emph{how}: it does not recover the generator $N$ or the Lyapunov structure $\dot{P} = NP + PN$.
Extending the mirror's estimation guarantees from the scalar distance to the full matrix $\dot{P}$ is an interesting open problem that would connect their framework to the Lyapunov inversion of \cref{prop:lyapunov-invert}.

\subsection{Holonomy for horizontal dynamics families (proved/conditional/conjectural)}\label{sec:holonomy-dynamics}

The holonomy obstruction (\cref{sec:fiber-bundle}) raises a concrete question: for which horizontal dynamics families is holonomy trivial, and for which is it nontrivial?
This distinction has practical bite: trivial holonomy means a globally consistent gauge is achievable \emph{in principle} (put differently: alignment is ``only'' statistical); nontrivial holonomy means that even perfect local alignment can accumulate global drift over cycles.

As summarized in \cref{tab:result-status}, this section contains:
\begin{enumerate}
\item proved local criteria and the $d = 2$ full restricted holonomy consequence,
\item a conditional theorem/proposition for general $d$,
\item and a conjectural generic full-holonomy extension for $d \ge 3$.
\end{enumerate}

In this section we make the mechanism explicit for two dynamics families.
The punchline is a commuting/non-commuting contrast:
polynomial dynamics use generators that commute along the trajectory (hence no curvature along that path, and no loops to begin with), while Laplacian-type generators can fail to commute; whenever the projected-commutator criterion is nonzero, curvature is positive and nontrivial holonomy becomes possible.

\textbf{Curvature from non-commuting generators.}
Consider horizontal dynamics $\dot{X} = M(X) X$ with $M(X)$ symmetric, so that $X^\top \dot{X} = X^\top M(X) X$ is symmetric and the dynamics are horizontal.
At two distinct times $t_1, t_2$, the generators $M_1 = M(X(t_1))$ and $M_2 = M(X(t_2))$ are both symmetric.
The commutator $[M_1, M_2] = M_1 M_2 - M_2 M_1$ is skew-symmetric (since transposing reverses the order).

The Lie bracket of the corresponding horizontal vector fields $\bar{\xi}_i(X) = M_i X$ (constant-coefficient extensions, horizontal since each $M_i$ is symmetric) is:
\begin{equation}
[\bar{\xi}_1, \bar{\xi}_2](X) = (M_2 M_1 - M_1 M_2) X = -[M_1, M_2] X
\end{equation}

Writing $S = -[M_1, M_2]$ (skew-symmetric), the bracket is $[\bar{\xi}_1, \bar{\xi}_2](X) = SX$.
This is \emph{not} automatically vertical: vertical directions are right-multiplications $X\Omega$, while $SX$ is a left-multiplication. They only coincide in special cases (roughly, when $S$ preserves $\col(X)$).

But for curvature we only need the \emph{vertical component} $[\bar{\xi}_1, \bar{\xi}_2]^{\mathcal{V}}$.
Projecting $Z = SX$ onto $\mathcal{V}_X = \{X\Omega : \Omega \in \so(d)\}$ gives $X\Omega^*$ where $\Omega^*$ solves the Lyapunov equation
\begin{equation}
G \Omega^* + \Omega^* G = 2 \skewop(X^\top Z) = 2 X^\top S X
\end{equation}
with $G = X^\top X$ (see \cref{app:vertical-projection}).
Since $G$ is positive definite, this has a unique solution. In the eigenbasis of $G$:
\begin{equation}
\Omega^*_{\iota\gamma} = \frac{2(X^\top S X)_{\iota\gamma}}{\lambda_\iota + \lambda_\gamma}
\end{equation}

So the vertical component vanishes iff the \emph{projected commutator} vanishes:
$X^\top [M_1, M_2] X = 0$.

\begin{proposition}[Generic Nonvanishing of Projected Skew Action]\label{prop:generic-projected-skew}
Let $n \ge d \ge 2$ and let $S \in \so(n)$ be nonzero.
Define
\begin{equation}
\mathcal{N}_S := \{X \in \RR_*^{n \times d} : X^\top S X = 0\}.
\end{equation}
Then $\mathcal{N}_S$ is a proper real-algebraic subset of $\RR_*^{n \times d}$.
In particular, it has empty interior and Lebesgue measure zero.
\end{proposition}

\begin{proof}
Each entry of $X^\top S X$ is polynomial in the entries of $X$, so $\mathcal{N}_S$ is algebraic.
It is proper because $S \neq 0$ implies there exist indices $i \neq j$ with $S_{ij} \neq 0$.
Choose $x_1 = e_i$ and $x_2 = e_j$; then $x_1^\top S x_2 = S_{ij} \neq 0$, so for $X_0 = [x_1, x_2, x_3, \ldots, x_d]$ with remaining columns chosen to make full column rank, $(X_0^\top S X_0)_{12} \neq 0$.
Hence $X_0$ is outside $\mathcal{N}_S$, so $\mathcal{N}_S$ is proper.
Proper real-algebraic subsets of Euclidean space have empty interior and measure zero.
\end{proof}

Thus $X^\top [M_1, M_2] X = 0$ is strictly weaker than $[M_1, M_2]=0$, and for fixed nonzero commutator it fails only on an algebraic exceptional set in $X$.
Intersecting that exceptional set with the interior of the valid domain $\mathcal{E}$ preserves this interpretation: on the interior model domain, the failure set remains proper.

Remembering \cref{prop:curv-spectral-gap}, by O'Neill's formula~\citep{oneill1966fundamental} the sectional curvature of the 2-plane spanned by the projections of $\bar{\xi}_1, \bar{\xi}_2$ in $\mathcal{B}$ is: $K(\xi_1, \xi_2) = \frac{3}{4} \|[\bar{\xi}_1, \bar{\xi}_2]^{\mathcal{V}}\|^2 / (\|\bar{\xi}_1\|^2 \|\bar{\xi}_2\|^2 - \langle \bar{\xi}_1, \bar{\xi}_2 \rangle^2)$.
The $1 / (\lambda_\iota + \lambda_\gamma)$ factors in $\Omega^*$ connect the curvature directly to the connection coefficients of the fiber bundle: the same denominators that amplify gauge sensitivity also amplify the vertical bracket.

Since $\RR_*^{n \times d}$ is flat (it is an open subset of Euclidean space), the base curvature arises \emph{entirely} from the vertical bracket.

\begin{proposition}[Curvature Criterion for Horizontal Dynamics]\label{prop:curvature-criterion}
For horizontal dynamics $\dot{X} = M(X) X$ with $M$ symmetric, the sectional curvature along the trajectory in $\mathcal{B}$ vanishes if and only if the projected commutator vanishes: $X^\top [M(X(t_1)), M(X(t_2))] X = 0$ for all $t_1, t_2$ along the trajectory.
A sufficient condition is that the full generators commute: $[M(X(t_1)), M(X(t_2))] = 0$.
When the projected commutator is nonzero, the curvature is strictly positive, with:
\begin{equation}
\|[\bar{\xi}_1, \bar{\xi}_2]^{\mathcal{V}}\|^2 = 4 \sum_{\iota < \gamma} \frac{\lambda_\iota \lambda_\gamma}{\lambda_\iota + \lambda_\gamma} [(U^\top [M_1, M_2] U)_{\iota\gamma}]^2
\end{equation}
where $P = U \Lambda U^\top$ is the eigendecomposition. The $1 / (\lambda_\iota + \lambda_\gamma)$ factors link the curvature directly to the connection coefficients of the fiber bundle.
\end{proposition}

\begin{proof}
The Lie bracket of the constant-coefficient horizontal extensions $\bar{\xi}_i(Y) = M_i Y$ is $[\bar{\xi}_1, \bar{\xi}_2](X) = -[M_1, M_2] X$ (standard commutator of linear vector fields).
The vertical projection of $SX$ (with $S = -[M_1, M_2]$ skew-symmetric) is $X\Omega^*$ where $\Omega^*$ solves $G\Omega^* + \Omega^* G = 2 X^\top S X$ (\cref{app:vertical-projection}).
This vanishes iff $X^\top [M_1, M_2] X = 0$.
The norm formula follows from $\|X\Omega^*\|^2 = \tr(\Omega^{*\top} G \Omega^*)$ evaluated in the eigenbasis of $G$ (see \cref{app:vertical-norm} for the computation).
\end{proof}

\textbf{Polynomial dynamics: trivial holonomy.}
For polynomial dynamics $\dot{X} = N(P) X$ with $N(P) = \sum_{k=0}^K \alpha_k P^k$, the induced $P$-dynamics are:
\begin{equation}
\dot{P} = N(P) P + P N(P) = 2 \sum_{k=0}^K \alpha_k P^{k+1}
\end{equation}
This is a polynomial in $P$ alone. Writing $P = U \Lambda U^\top$ in its eigendecomposition, every power $P^k = U \Lambda^k U^\top$ shares the same eigenvectors $U$.
Therefore $\dot{P} = U (2 \sum \alpha_k \Lambda^{k+1}) U^\top$ also has eigenvectors $U$: the eigenvectors of $P$ are stationary under polynomial dynamics, provided the spectrum of $P(0)$ is simple (all eigenvalues distinct).

The trajectory $P(t)$ lies on a $d$-dimensional submanifold of $\mathcal{B}$ parameterized by the eigenvalues $(\lambda_1(t), \ldots, \lambda_d(t))$ alone, with each eigenvalue evolving independently:
\begin{equation}
\dot{\lambda}_\iota = 2 \sum_{k=0}^K \alpha_k \lambda_\iota^{k+1}
\end{equation}

Two consequences follow immediately.

\begin{proposition}[Polynomial Dynamics: Trivial Holonomy]\label{prop:poly-trivial-holonomy}
Polynomial dynamics $\dot{X} = N(P) X$ with $N(P) = \sum \alpha_k P^k$ and simple initial spectrum ($\lambda_1(0) > \lambda_2(0) > \cdots > \lambda_d(0) > 0$) have trivial holonomy, in the following concrete sense:
\begin{enumerate}
\item The generators $N(P(t_1))$ and $N(P(t_2))$ commute for all $t_1, t_2$ (they share eigenvectors), so the curvature \emph{along this trajectory} vanishes by \cref{prop:curvature-criterion}.
\item Each eigenvalue satisfies a 1D autonomous ODE, so non-constant eigenvalue trajectories are monotone. In particular, the induced $P(t)$ cannot trace a nontrivial closed loop in $\mathcal{B}$.
\end{enumerate}
\end{proposition}

\begin{proof}
For (1): since the eigenvectors of $P$ are stationary, $N(P(t)) = U (\sum \alpha_k \Lambda(t)^k) U^\top$ for all $t$, with the same $U$. Any two diagonal matrices commute, so $[N(P(t_1)), N(P(t_2))] = U [\sum \alpha_k \Lambda(t_1)^k, \sum \alpha_k \Lambda(t_2)^k] U^\top = 0$.

For (2): each $\lambda_\iota(t)$ satisfies the autonomous ODE $\dot{\lambda}_\iota = f(\lambda_\iota)$ with $f(x) = 2 \sum \alpha_k x^{k+1}$.
By uniqueness of ODE solutions, if $\lambda_\iota(t_0) = \lambda_\iota(t_1)$ for $t_0 \neq t_1$, then $\lambda_\iota$ must be constant.
Non-constant solutions therefore cannot be periodic: they cannot return to a previous value. Hence no nontrivial closed loop in eigenvalue space is possible. As a consequence, no nontrivial closed loop in $\mathcal{B}$ is possible.
\end{proof}

\textbf{What polynomial dynamics look like for nodes.}
The statement that eigenvectors of $P$ are stationary deserves both rigorous justification and concrete interpretation.

\emph{Why eigenvectors are fixed.}
Write $P(t) = U(t) \Lambda(t) U(t)^\top$ and let $\Omega = U^\top \dot{U} \in \so(d)$ be the eigenvector rotation rate.
Differentiating $P = U \Lambda U^\top$ and projecting into the eigenbasis gives $(U^\top \dot{P} U)_{\iota\gamma} = \Omega_{\iota\gamma}(\lambda_\gamma - \lambda_\iota)$ for $\iota \neq \gamma$.
Since $\dot{P} = 2 \sum \alpha_k P^{k+1}$ is diagonal in the eigenbasis, the left side vanishes, so $\Omega_{\iota\gamma} (\lambda_\gamma - \lambda_\iota) = 0$ for all $\iota \neq \gamma$.
It remains to show that eigenvalues remain distinct.
All $d$ eigenvalues satisfy the same autonomous ODE $\dot{\lambda} = f(\lambda)$ with $f(x) = 2 \sum \alpha_k x^{k+1}$, differing only in initial conditions.
If $\lambda_\iota(0) \neq \lambda_\gamma(0)$, then by uniqueness of ODE solutions $\lambda_\iota(t) \neq \lambda_\gamma(t)$ for all $t$: if they ever coincided, they would be the same solution, contradicting their distinct initial conditions.
Therefore $\lambda_\gamma - \lambda_\iota \neq 0$, which forces $\Omega = 0$: the eigenvectors of $P$ are genuinely stationary, not merely stationary to first order.

The simple spectrum assumption is essential: if $\lambda_\iota(0) = \lambda_\gamma(0)$ for some $\iota \neq \gamma$, the eigenbasis within the repeated eigenspace is not unique, and polynomial dynamics preserve this degeneracy (the repeated eigenvalues evolve identically, maintaining the degeneracy for all time).
In this case, eigenvector ``stationarity'' is ill-defined because the basis itself is non-unique.
For block models with equal block sizes, for instance, repeated eigenvalues arise from the model symmetry.
This is a measure-zero condition in the space of initial conditions.

\emph{What this means for nodes.}
In the canonical gauge $X = U \Lambda^{1/2}$, node $i$ has latent position $x_i(t) = (U_{i1} \sqrt{\lambda_1(t)}, \ldots, U_{id} \sqrt{\lambda_d(t)})$.
The loading $U_{i\iota}$ (which can be read as the ``membership weight'' of node $i$ in eigenvector direction $\iota$) is constant; what evolves is the scale $\sqrt{\lambda_\iota(t)}$ of each direction.

Since the eigenvalue ODE $\dot{\lambda} = f(\lambda)$ is nonlinear for $K \ge 1$, eigenvalues with different magnitudes evolve at different rates, and the ratios $\lambda_\iota(t) / \lambda_\gamma(t)$ change over time.
The effect on the node positions is an \emph{anisotropic rescaling} of the latent space along fixed axes: in $d = 2$, the point cloud deforms as if inscribed in an ellipse whose axes do not rotate but whose eccentricity changes.
Node directions in $\RR^d$ do change (they are not merely rescaled in magnitude), but the change is tightly constrained: all deformation is captured by the $d$ scalar functions $\lambda_\iota(t)$.

The exception is linear dynamics ($K = 0$, $f(\lambda) = 2 \alpha_0 \lambda$), where $\lambda_\iota(t) = \lambda_\iota(0) e^{2 \alpha_0 t}$ and all eigenvalues grow or decay at the same exponential rate.
In this case the ratios are constant, the point cloud is rescaled isotropically, and node angular positions in $\RR^d$ are truly fixed.
For quadratic or higher dynamics ($K \ge 1$), the nonlinearity of $f$ causes larger eigenvalues to evolve differently from smaller ones, producing genuine reshaping of the point cloud's geometry while preserving the eigenvector axes that define community alignment.
Note, however, that preserving the axes does not guarantee preserving community \emph{separability}: if eigenvalue ratios diverge (e.g., $\lambda_d(t) \to 0$ while $\lambda_1$ grows), the point cloud collapses onto a lower-dimensional subspace, and communities that were distinguishable along the $\lambda_d$ direction may become indistinguishable.

\textbf{Laplacian dynamics: nontrivial holonomy.}
Graph Laplacian dynamics $\dot{X} = -LX$ with $L = D - P$, where $D = \diag(P\mathbf{1})$, provide a natural example of horizontal dynamics with nontrivial holonomy.
The generator $-L = P - D$ is symmetric, so the dynamics are horizontal.
But $L$ is \emph{not} a polynomial in $P$: the degree matrix $D$ depends on the row sums of $P$, which mix eigenvector information that a polynomial in $P$ cannot access.

There is a notable exception: if $D$ is a scalar multiple of the identity ($D = cI$, i.e., the graph is \emph{regular} with all node degrees equal), then $L = cI - P$ is affine in $P$, and the Laplacian dynamics reduce to the polynomial family.
Regular graphs (and more generally, configurations with constant row sums of $P$, i.e.\ $P\mathbf{1} = c\mathbf{1}$) thus have trivial holonomy under Laplacian dynamics.
The nontrivial holonomy results below require that $D$ is \emph{not} scalar: structural heterogeneity in the graph is the source of holonomy.

The $P$-dynamics under Laplacian flow are:
\begin{equation}
\dot{P} = -LP - PL = -(D - P)P - P(D - P) = 2P^2 - DP - PD
\end{equation}
The term $2P^2$ preserves eigenvectors of $P$ (it is a polynomial in $P$), but $DP + PD$ does not: in the eigenbasis $P = U \Lambda U^\top$, the degree matrix $D = \diag(U \Lambda U^\top \mathbf{1})$ is diagonal in the \emph{node} basis, not the eigenbasis.
Thus $U^\top D U$ is generally a full matrix, and $\dot{P}$ has eigenvector components transverse to those of $P$.

\emph{Consequence:} the eigenvectors of $P$ rotate under Laplacian dynamics, the trajectory sweeps out area in $\mathcal{B}$ transverse to the eigenvalue-parameterized submanifold, and \cref{prop:curvature-criterion} guarantees positive curvature along any portion of the trajectory where $[L(P(t_1)), L(P(t_2))] \neq 0$.

\textbf{What Laplacian dynamics look like for nodes.}
The contrast with polynomial dynamics is sharp at the node level.
Under Laplacian diffusion, each node's velocity $\dot{x}_i = -\sum_j L_{ij} x_j$ is a weighted average of the positions of its neighbors (with weights from $P$), minus a restoring term from the node's own degree.
This acts like heat diffusion on the latent positions: high-degree nodes experience stronger mean-reversion, while the diffusion couples all positions through the network structure.
Because the coupling depends on the \emph{row sums} of $P$, mixing contributions from all eigenvector directions, the eigenvectors of $P$ rotate over time.
In the ellipse analogy for $d = 2$: not only does the eccentricity change, but the axes themselves rotate.
The community structure is no longer static; nodes' membership weights across latent dimensions genuinely evolve, creating the richer dynamics that produce nontrivial holonomy.

\begin{proposition}[Laplacian Dynamics: Local Commutator Criterion for Nontrivial Holonomy]\label{prop:laplacian-holonomy}
Let $X(t)$ solve Laplacian dynamics $\dot{X} = -L(P) X$, with $P = XX^\top$ and $L = D - P$.
Suppose there exists $t_* \in [0,T)$ such that
\begin{equation}
C_* := X(t_*)^\top [L(t_*), \dot{L}(t_*)] X(t_*) \neq 0.
\end{equation}
Then there exists $\epsilon_0 > 0$ such that for all $0 < |h| < \epsilon_0$,
\begin{equation}
X(t_*)^\top [L(t_*), L(t_* + h)] X(t_*) \neq 0.
\end{equation}
Consequently, by \cref{prop:curvature-criterion}, the sectional curvature along the trajectory is strictly positive on that local two-time span.
In particular, there exist contractible loops in a neighborhood of that region whose restricted holonomy is nontrivial.
\end{proposition}

\begin{proof}
Since $L$ is differentiable in $t$, Taylor expansion gives
$L(t_* + h) = L(t_*) + h \dot{L}(t_*) + O(h^2)$.
Therefore
$[L(t_*), L(t_* + h)] = h [L(t_*), \dot{L}(t_*)] + O(h^2)$.
Left and right multiplication by the fixed matrix $X(t_*)^\top$ and $X(t_*)$ yields
\begin{equation}
X(t_*)^\top [L(t_*), L(t_* + h)] X(t_*) = h C_* + O(h^2).
\end{equation}
Because $C_* \neq 0$, continuity implies there exists $\epsilon_0 > 0$ such that this expression is nonzero for all $0 < |h| < \epsilon_0$.

Apply \cref{prop:curvature-criterion} with $M_1 = -L(t_*)$ and $M_2 = -L(t_* + h)$ (both symmetric): nonzero projected commutator implies strictly positive sectional curvature for the corresponding two-time span.
The final holonomy claim follows from Ambrose--Singer: nonzero curvature implies nontrivial restricted holonomy, hence existence of contractible loops with non-identity holonomy.
\end{proof}

The criterion above is checkable pointwise along a trajectory. It is also generic at initialization.

\begin{proposition}[Genericity of the Local Commutator Condition at Initialization]\label{prop:laplacian-generic-init}
Fix $n \ge 3$ and $2 \le d \le n$, and define
\begin{equation}
\Psi(X) := X^\top [L(X), \dot{L}(X)] X
\end{equation}
with $P = XX^\top$, $D = \diag(P\mathbf{1})$, $L = D - P$, $\dot{P} = 2P^2 - DP - PD$, and $\dot{L} = \diag(\dot{P}\mathbf{1}) - \dot{P}$.
Then
\begin{equation}
\mathcal{Z} := \{X \in \RR_*^{n \times d} : \Psi(X) = 0\}
\end{equation}
is a proper real-algebraic subset of $\RR_*^{n \times d}$.
Consequently, $\mathcal{Z}$ has empty interior and Lebesgue measure zero; equivalently, for generic initial conditions, the hypothesis of \cref{prop:laplacian-holonomy} holds at $t_* = 0$.
Restricting to valid RDPG states, the same conclusion holds after intersecting with $\mathcal{E}$: this remains a proper exceptional set in the interior model domain.
\end{proposition}

\begin{proof}
Each entry of $\Psi(X)$ is a polynomial in the entries of $X$ (all operations are sums and products), so $\mathcal{Z}$ is algebraic.
To show it is proper, it suffices to exhibit one full-rank $X$ with $\Psi(X) \neq 0$.

For $n = 3$, $d = 2$, take
\begin{equation}
X_0 = \frac{1}{3}\begin{pmatrix} 1 & 1 \\ 2 & 1 \\ 2 & 2 \end{pmatrix}.
\end{equation}
A direct substitution into the formulas above gives
$(\Psi(X_0))_{12} = 2/3^7 \neq 0$,
hence $\Psi(X_0) \neq 0$.

For general $n \ge 3$ and $2 \le d \le n$, embed this $2$-dimensional witness into the first two columns and first three rows, set remaining entries to zero, and then add an arbitrarily small perturbation in columns $3, \ldots, d$ making the matrix full rank.
Since $\Psi$ is continuous and nonzero at the embedded witness, for sufficiently small perturbation the value remains nonzero.
Therefore a full-rank witness exists for every $n \ge 3$ and $2 \le d \le n$, so $\mathcal{Z}$ is a proper algebraic subset.

A proper real-algebraic subset of Euclidean space has empty interior and Lebesgue measure zero.
\end{proof}

For $d = 2$, the result is particularly sharp.

\begin{corollary}[Dimension 2: Full Restricted Holonomy from One Nonzero Curvature Value]\label{cor:d2-holonomy}
Let $d = 2$ and assume the hypothesis of \cref{prop:laplacian-holonomy} holds at some time $t_*$.
Then the restricted holonomy group of the bundle is $\mathrm{SO}(2)$.
Equivalently, for each angle $\phi \in [0, 2\pi)$ there exists a contractible loop in $\mathcal{B}$ whose horizontal transport has holonomy rotation by $\phi$.
\end{corollary}

\begin{proof}
The Lie algebra $\so(2)$ is one-dimensional, spanned by the single generator $J = \bigl(\begin{smallmatrix} 0 & -1 \\ 1 & 0 \end{smallmatrix}\bigr)$.
By \cref{prop:laplacian-holonomy}, the curvature 2-form produces a nonzero element of $\so(2)$.
A single nonzero element spans the one-dimensional Lie algebra, so by the Ambrose--Singer theorem, the holonomy algebra is all of $\so(2)$, and the restricted holonomy group (holonomy of contractible loops) is $\mathrm{SO}(2)$.
\end{proof}

This means that for rank-2 latent spaces under Laplacian dynamics, the obstruction is not merely a discrete sign flip (as in the simplest cases) but a continuous rotation by an arbitrary angle: you can get \emph{any} amount of gauge drift, depending on the loop. From an alignment point of view, this is about as bad as it gets.

\begin{remark}
\textbf{Heuristic quantitative estimate (dimension 2).}
For $d = 2$, the holonomy around a loop $\gamma$ in $\mathcal{B}$ is often approximated by an area integral of sectional curvature:
$\phi(\gamma) \approx \iint_\Sigma K \, dA$
where $\Sigma$ is a surface bounded by $\gamma$, $K$ is the sectional curvature, and $dA$ is the area element.
This is an intuition aid, not a theorem used elsewhere in the manuscript.
Combined with~\citet{massart2019curvature}, it suggests that trajectories passing near rank-deficient states (small trailing eigenvalues) may accumulate larger holonomy.
\end{remark}

\textbf{Higher dimensions ($d \ge 2$).}
For general $d$, the Lie algebra $\so(d)$ has dimension $d(d-1)/2$.
The curvature 2-form produces holonomy-algebra elements $\Omega^* \in \so(d)$ with
$\Omega^*_{\iota\gamma} = 2(X^\top [M_1, M_2] X)_{\iota\gamma} / (\lambda_\iota + \lambda_\gamma)$.
As the trajectory visits different states, the \emph{direction} of these elements in $\so(d)$ can change, and the span can grow.

\begin{theorem}[Full Restricted Holonomy from a Curvature-Span Condition]\label{thm:full-holonomy-conditional}
Let $d \ge 2$ and fix a base point $p \in \mathcal{B}$.
Suppose the Lie algebra generated by curvature values reachable from $p$ spans all of $\so(d)$, i.e.
\begin{equation}
\spn\{ \Ad_{\tau_{q \to p}} (\Omega_q(u, v)) : q \text{ horizontally reachable from } p,~u,v \in \mathcal{H}_q \} = \so(d).
\end{equation}
Then the restricted holonomy group at $p$ is:
$\Hol_p^0 = \mathrm{SO}(d)$.
\end{theorem}

\begin{proof}
By the Ambrose--Singer theorem, the Lie algebra of $\Hol_p^0$ is exactly the span above.
Under the hypothesis, $\Lie(\Hol_p^0) = \so(d)$.
The group $\Hol_p^0$ is connected and lies in $O(d)$; hence it lies in the identity component $\mathrm{SO}(d)$.
A connected Lie subgroup of $\mathrm{SO}(d)$ with Lie algebra $\so(d)$ must be $\mathrm{SO}(d)$ itself.
\end{proof}

The theorem isolates the exact missing ingredient for $d \ge 3$: a rank condition on curvature-generated directions.

\begin{proposition}[Finite-Time Rank Criterion at a Fixed Base Point]\label{prop:finite-time-rank}
Let $X(t)$ solve $\dot{X} = -L(P) X$ and fix $t_0$ such that $P(t_0)$ has rank $d$.
Suppose there exist times $s_1, \ldots, s_m$ with $m = d(d-1)/2$ such that the skew matrices
\begin{equation}
B_a := X(t_0)^\top [L(t_0), L(s_a)] X(t_0) \in \so(d), \quad a = 1, \ldots, m,
\end{equation}
are linearly independent.
Then, at the base point $p = [X(t_0)]$, the restricted holonomy group is $\mathrm{SO}(d)$.
\end{proposition}

\begin{proof}
By \cref{prop:curvature-criterion} and \cref{app:vertical-projection}, each $B_a$ determines a curvature value at the same base point $p$ through the Lyapunov map
$B_a \mapsto \Omega_a^*$,
with componentwise scaling by positive factors $1/(\lambda_\iota + \lambda_\gamma)$ (evaluated at $t_0$).
This linear map on $\so(d)$ is invertible, so linear independence of $(B_a)$ implies linear independence of $(\Omega_a^*)$.
Since $m = \dim(\so(d))$, these curvature values span $\so(d)$ at $p$.
Therefore the curvature-span hypothesis of \cref{thm:full-holonomy-conditional} holds (with $q = p$, so no transport issue), and $\Hol_p^0 = \mathrm{SO}(d)$.
\end{proof}

\begin{remark}
\textbf{Heuristic genericity note.}
For Laplacian dynamics with $n \gg d$, one expects the degree matrices at different times to make the finite-time rank criterion hold on most trajectories.
Turning this into a theorem for all $d \ge 3$ requires a transversality argument controlling symmetry strata.
We therefore keep the fully generic statement as a conjecture.
\end{remark}

\begin{conjecture}[Full Holonomy in Higher Dimensions]\label{conj:full-holonomy}
For $d \ge 3$, $n > d$, and initial conditions with $D(P(0))$ not scalar (non-regular graph), the restricted holonomy group of Laplacian dynamics $\dot{X} = -L(P) X$ is $\mathrm{SO}(d)$.
\end{conjecture}

A proof of \cref{conj:full-holonomy} would amount to proving the rank criterion above generically.
Equivalently, one must show that the projected commutators
$X^\top [L(P(t_1)), L(P(t_2))] X$
span $\so(d)$ for generic trajectories.
The main obstruction is structural symmetry: if the degree operator $P \mapsto D(P)$ preserves a proper invariant subspace in $\so(d)$, the generated holonomy algebra can collapse to a proper subalgebra.
Regular/equitable families are canonical examples of such degeneracy.

The same argument applies to message-passing dynamics $\dot{x}_i = \sum_j P_{ij} g(x_i, x_j)$ with symmetric generators that are not polynomials in $P$.

\textbf{Summary:} the holonomy obstruction creates a clear contrast among horizontal dynamics families.
Polynomial dynamics, which operate on the spectral structure of $P$ through commuting generators, preserve eigenvectors and have trivial holonomy: global gauge consistency is achievable without topological obstruction.
Laplacian and message-passing dynamics, which mix spectral and spatial structure through non-commuting generators, can produce nontrivial holonomy whenever the projected-commutator criterion of \cref{prop:curvature-criterion} is triggered: gauge drift then accumulates even along horizontal trajectories, and no local alignment procedure can achieve global consistency over cycles.

This distinction has practical implications: for polynomial dynamics, the constructive alignment problem (\cref{sec:constructive}) is purely a statistical challenge (overcoming Bernoulli noise); for Laplacian dynamics, it is simultaneously a statistical \emph{and} topological challenge.

\begin{remark}
\textbf{Ambrose--Singer perspective (a quick intuition pump).}
The relationship between local curvature and global holonomy is formalized by the Ambrose--Singer theorem~\citep{kobayashi1963foundations}: the Lie algebra $\hol_p$ of the restricted holonomy group at $p$ is generated by curvature endomorphisms $\Omega_q(u, v)$ evaluated at points $q$ reachable from $p$ by horizontal curves, after parallel-transporting them back to $p$.

For $d = 2$, $\so(2)$ has no proper nonzero subalgebras, so a single nonzero curvature element forces the full restricted holonomy $\Hol_p^0 = \mathrm{SO}(2)$ (as used above).

For $d \ge 3$, one curvature value gives only one direction inside a much larger algebra.
Full holonomy requires a bracket-generating phenomenon: curvature directions encountered along the trajectory must not all live inside a commuting subalgebra.
Since the degree matrix $D(X)$ depends nonlinearly on the latent positions, the orientation of the induced curvature elements typically rotates as the trajectory evolves, which is exactly the mixing mechanism behind \cref{conj:full-holonomy}.

Finally, the spectral gap appears again: as shown in \cref{prop:curvature-criterion}, the curvature magnitude is weighted by $1/(\lambda_\iota + \lambda_\gamma)$.
So a small spectral gap not only amplifies statistical noise (via ASE error bounds) but also makes holonomy effects larger, because the curvature-generated gauge rotations blow up in the ill-conditioned directions.
\end{remark}



\section{Information-theoretic limits}\label{sec:info-theoretic}

The geometric obstructions of \cref{sec:obstructions} (gauge freedom, curvature, holonomy), the trajectory recovery difficulties of \cref{sec:trajectory-problem}, and the dynamics analysis of \cref{sec:dynamics} constrain what is learnable in principle.
We complement these with \emph{statistical} obstructions: given $T$ noisy adjacency matrices, how accurately can we estimate the parameters governing the dynamics, regardless of the estimation method?

The answer reveals a duality: the same spectral gap $\lambda_d$ that controls geometric difficulty also controls statistical difficulty, so the two obstructions reinforce each other.

\subsection{Fisher information for dynamics parameters}

Consider a parametric dynamics family with parameter $\theta \in \RR^k$ generating a trajectory $P(t; \theta)$.
We observe adjacency matrices $\{A^{(t)}\}_{t=0}^T$ with $A^{(t)} \sim \text{Bernoulli}(P(t; \theta))$, and all entries conditionally independent given $P(t; \theta)$.
The log-likelihood factorizes:
\begin{equation}
\ell(\theta) = \sum_{t=0}^T \sum_{i < j} \bigl[A_{ij}^{(t)} \log P_{ij}(t; \theta) + (1 - A_{ij}^{(t)}) \log(1 - P_{ij}(t; \theta))\bigr]
\end{equation}

The Fisher information matrix $\mathcal{I}(\theta) \in \RR^{k \times k}$ has entries:
\begin{equation}\label{eq:fisher}
\mathcal{I}(\theta)_{ab} = \sum_{t=0}^T \sum_{i < j} \frac{\partial P_{ij}}{\partial \theta_a} \frac{\partial P_{ij}}{\partial \theta_b} \cdot \frac{1}{P_{ij}(1 - P_{ij})}
\end{equation}
where $P_{ij} = P_{ij}(t; \theta)$ and the derivatives $\partial P_{ij} / \partial \theta_a$ capture how perturbations in the dynamics parameters propagate to the observations.

\subsection{Sensitivity propagation through the Lyapunov equation}

The key quantity in \cref{eq:fisher} is the sensitivity $\partial P_{ij} / \partial \theta_a$.
For polynomial dynamics $N(\theta) = \sum_{k=0}^K \theta_k P^k$, the $P$-dynamics are:
$\dot{P} = 2 \sum_{k=0}^K \theta_k P^{k+1}$.
A crucial simplification comes from \cref{prop:poly-trivial-holonomy}: the eigenvectors of $P$ are stationary under polynomial dynamics.
Writing $P = U \Lambda U^\top$ with $\Lambda = \diag(\lambda_1, \ldots, \lambda_d)$ and $U \in \RR^{n \times d}$ fixed, the sensitivity $\partial P / \partial \theta_a = U (\partial \Lambda / \partial \theta_a) U^\top$ is diagonal in the eigenbasis.

Each eigenvalue satisfies $\dot{\lambda}_\iota = 2 \sum_{k=0}^K \theta_k \lambda_\iota^{k+1}$, so:
\begin{equation}
\frac{\partial}{\partial t} \frac{\partial \lambda_\iota}{\partial \theta_a} = 2 \lambda_\iota^{a+1} + \frac{\partial \dot{\lambda}_\iota}{\partial \lambda_\iota} \cdot \frac{\partial \lambda_\iota}{\partial \theta_a}
\end{equation}
where $\partial \dot{\lambda}_\iota / \partial \lambda_\iota = 2 \sum_{k=0}^K \theta_k (k+1) \lambda_\iota^k$ is the linearized eigenvalue dynamics.
This is a scalar linear ODE for each pair $(\iota, a)$, with forcing $2 \lambda_\iota^{a+1}$ and feedback through the linearized dynamics.

The resulting sensitivity of the observed probabilities is:
\begin{equation}\label{eq:poly-sensitivity}
\frac{\partial P_{ij}}{\partial \theta_a} = \sum_{\iota=1}^d U_{i\iota} U_{j\iota} \frac{\partial \lambda_\iota}{\partial \theta_a}
\end{equation}
Each node pair $(i,j)$ receives signal from all $d$ eigenvalue sensitivities, weighted by the eigenvector loadings $U_{i\iota} U_{j\iota}$.

\textbf{Example: linear dynamics ($K = 0$).}
For $\dot{X} = \alpha_0 X$ (single parameter), the eigenvalue dynamics are $\dot{\lambda}_\iota = 2 \alpha_0 \lambda_\iota$, giving $\lambda_\iota(t) = \lambda_\iota(0) e^{2 \alpha_0 t}$.
The sensitivity is $\partial \lambda_\iota / \partial \alpha_0 = 2 t \lambda_\iota(t)$, and therefore:
\begin{equation}
\frac{\partial P_{ij}}{\partial \alpha_0} = 2 t \sum_\iota U_{i\iota} \lambda_\iota(t) U_{j\iota} = 2 t P_{ij}(t)
\end{equation}
The Fisher information at a single snapshot $t$ is:
\begin{equation}
\mathcal{I}_t(\alpha_0) = 4 t^2 \sum_{i < j} \frac{P_{ij}(t)^2}{P_{ij}(t)(1 - P_{ij}(t))} = 4 t^2 \sum_{i < j} \frac{P_{ij}(t)}{1 - P_{ij}(t)}
\end{equation}
Summing over $T$ equally-spaced snapshots gives $\mathcal{I}(\alpha_0) = 4 \sum_{t=0}^T t^2 \sum_{i<j} P_{ij}(t) / (1 - P_{ij}(t))$, which scales as $O(T^3 \cdot n^2)$ when the edge probabilities are bounded away from 0 and 1.

\textbf{Higher-degree terms and the spectral gap.}
For $a \ge 1$, the forcing $2 \lambda_\iota^{a+1}$ is small for eigenvalues near the spectral gap: $2 \delta^{a+1}$ versus $2 \lambda_1^{a+1}$.
This means the sensitivity $\partial \lambda_d / \partial \theta_a$ is small relative to $\partial \lambda_1 / \partial \theta_a$, so the Fisher information for higher-degree parameters $\theta_a$ receives proportionally less signal from the small-eigenvalue directions.
Concretely, the contribution of the $\lambda_d$ direction to $\mathcal{I}(\theta)_{ab}$ scales as $\delta^{a+b+2}$, while the $\lambda_1$ direction contributes $\lambda_1^{a+b+2}$.

\subsection{General dynamics and the Lyapunov amplification}

The polynomial case is clean because eigenvectors are fixed, confining all sensitivity to the diagonal of the eigenbasis.
For general symmetric dynamics $\dot{X} = M(P) X$ with $M$ symmetric but not a polynomial in $P$, eigenvectors rotate and the sensitivity $\partial P / \partial \theta$ has both diagonal and off-diagonal components in the eigenbasis.

In this setting, the Lyapunov equation $\dot{P} = MP + PM$ must be inverted to recover $M$ from $\dot{P}$.
In the eigenbasis of $P$, the inversion acts componentwise: $M_{\iota\gamma} = \dot{P}_{\iota\gamma} / (\lambda_\iota + \lambda_\gamma)$.
The factor $1 / (\lambda_\iota + \lambda_\gamma)$ amplifies noise in the off-diagonal components, and is the same factor that appears in the connection 1-form (\cref{sec:fiber-bundle}) and governs curvature.
This yields a direct correspondence: directions in parameter space that are hard geometrically (small $\lambda_\iota + \lambda_\gamma$ produces large curvature and large connection coefficients) are also hard statistically (small $\lambda_\iota + \lambda_\gamma$ amplifies estimation noise).

For Laplacian dynamics, the situation is compounded by holonomy.
The eigenvector rotation contributes additional degrees of freedom to the sensitivity, but this information is entangled with gauge degrees of freedom that may fail to be locally resolvable when the nonzero-curvature criterion of \cref{prop:laplacian-holonomy} is met.

\subsection{The statistical-geometric duality}

Combining the geometric and statistical pictures yields a unified obstruction.
We first state the result for polynomial dynamics, then explain the broader duality.

\begin{proposition}[Cram\'er--Rao Bound for Polynomial Dynamics]\label{prop:cramer-rao}
For polynomial dynamics of degree $K$ with parameter $\theta = (\theta_0, \ldots, \theta_K) \in \RR^{k}$ ($k = K + 1$), assume:
\begin{enumerate}
\item $\theta$ belongs to an open parameter set $\Theta \subset \RR^k$, and the true parameter is an interior point.
\item For each $t, i, j$, the map $\theta \mapsto P_{ij}(t; \theta)$ is $C^1$, with $0 < P_{ij}(t; \theta) < 1$ on a neighborhood of the true parameter.
\item The score is integrable and differentiation can be interchanged with expectation (standard Fisher regularity), so the information identity is valid.
\item $\mathcal{I}(\theta)$ is positive definite on the identifiable parameter subspace.
\end{enumerate}

Then any unbiased estimator $\hat{\theta}$ satisfies:
\begin{equation}
\EE[\|\hat{\theta} - \theta\|^2] \ge \tr(\mathcal{I}(\theta)^{-1})
\end{equation}
The Fisher information matrix $\mathcal{I}(\theta)$ has entries:
\begin{equation}
\mathcal{I}(\theta)_{ab} = \sum_{t=0}^T \sum_{\iota, \gamma = 1}^d \frac{\partial \lambda_\iota}{\partial \theta_a} \frac{\partial \lambda_\gamma}{\partial \theta_b} \cdot C_{\iota\gamma}(t)
\end{equation}
where $C_{\iota\gamma}(t) = \sum_{i < j} U_{i\iota} U_{j\iota} U_{i\gamma} U_{j\gamma} / [P_{ij}(t)(1 - P_{ij}(t))]$.
The diagonal entries $C_{\iota\iota}(t) = \mathcal{W}_\iota(t)$ measure the information content of eigenvalue direction $\iota$ at time $t$; the off-diagonal entries $C_{\iota\gamma}(t)$ for $\iota \neq \gamma$ capture cross-eigenvalue correlations.
A baseline lower bound on each parameter's variance follows from the matrix inequality $(\mathcal{I}^{-1})_{aa} \ge 1 / \mathcal{I}_{aa}$:
\begin{equation}
\mathrm{Var}(\hat{\theta}_a) \ge \frac{1}{\mathcal{I}_{aa}} = \frac{1}{\sum_t \sum_\iota \bigl(\frac{\partial \lambda_\iota}{\partial \theta_a}\bigr)^2 \mathcal{W}_\iota(t) + \text{cross terms}}
\end{equation}
Non-zero cross-eigenvalue terms arising from structured graphs (where Fisher weights correlate with eigenvector components) can only \emph{increase} the true Cram\'er--Rao bound above this baseline.
\end{proposition}

\begin{proof}
By \cref{eq:poly-sensitivity}, the sensitivity of $P_{ij}$ decomposes as $\partial P_{ij} / \partial \theta_a = \sum_\iota U_{i\iota} U_{j\iota} (\partial \lambda_\iota / \partial \theta_a)$.
Substituting into \cref{eq:fisher}:
\begin{equation}
\mathcal{I}(\theta)_{ab} = \sum_t \sum_{i < j} \biggl[\sum_\iota U_{i\iota} U_{j\iota} \frac{\partial \lambda_\iota}{\partial \theta_a}\biggr] \biggl[\sum_\gamma U_{i\gamma} U_{j\gamma} \frac{\partial \lambda_\gamma}{\partial \theta_b}\biggr] \cdot \frac{1}{P_{ij}(1 - P_{ij})}
\end{equation}
Expanding the product over $\iota, \gamma$ yields:
\begin{equation}
\mathcal{I}(\theta)_{ab} = \sum_t \sum_{\iota, \gamma} \frac{\partial \lambda_\iota}{\partial \theta_a} \frac{\partial \lambda_\gamma}{\partial \theta_b} \underbrace{\sum_{i<j} \frac{U_{i\iota} U_{j\iota} U_{i\gamma} U_{j\gamma}}{P_{ij}(1 - P_{ij})}}_{C_{\iota\gamma}(t)}
\end{equation}

\textbf{On the cross-eigenvalue terms.}
Column orthogonality of $U$ gives $\sum_i U_{i\iota} U_{i\gamma} = \delta_{\iota\gamma}$, but the Fisher weights $1/[P_{ij}(1-P_{ij})]$ are not uniform.
In structured RDPGs (e.g., stochastic block models with equal-sized blocks), these weights are constant within blocks and the blocks are defined by the eigenvectors themselves, so the weights correlate with eigenvector components and cross terms $C_{\iota\gamma}$ for $\iota \neq \gamma$ can be nonzero (and are nonzero in explicit blockmodel configurations).
The Fisher information matrix is therefore \emph{not} diagonal in the eigenvalue basis for structured graphs.

However, the spectral gap scaling is unaffected: the eigenvalue sensitivity $\partial \lambda_\iota / \partial \theta_a$ is determined by the scalar ODE $\dot{\lambda}_\iota = 2 \sum_k \theta_k \lambda_\iota^{k+1}$.
The forcing term $2 \lambda_\iota^{a+1}$ scales as $\delta^{a+1}$ for $\iota = d$, so the $\lambda_d$ direction contributes $O(\delta^{a+b+2})$ to $\mathcal{I}_{ab}$ regardless of whether this contribution arises from diagonal or cross terms.
When $\delta \to 0$, the Fisher information matrix becomes ill-conditioned: the rows and columns involving high powers of $\delta$ shrink, making $\tr(\mathcal{I}^{-1})$ diverge.

The baseline bound $(\mathcal{I}^{-1})_{aa} \ge 1/\mathcal{I}_{aa}$ is immediate from Cauchy--Schwarz: $1 = (e_a^\top e_a)^2 \le (e_a^\top \mathcal{I} e_a)(e_a^\top \mathcal{I}^{-1} e_a) = \mathcal{I}_{aa} (\mathcal{I}^{-1})_{aa}$.
This bound is tight when $\mathcal{I}$ is diagonal and loose when off-diagonal terms are large (i.e., when parameter correlations from structured graphs make estimation strictly harder than the diagonal analysis suggests).
\end{proof}

If $\mathcal{I}(\theta)$ is singular, the same derivation applies on the identifiable subspace using the Moore--Penrose pseudoinverse $\mathcal{I}^+$.

The bound reveals three scaling regimes.
Each snapshot provides $O(n^2)$ conditionally independent Bernoulli observations, so the Fisher information grows as $n^2$.
Information accumulates linearly in $T$, assuming the dynamics produce sufficient variation in $P$ across time.
The spectral gap $\delta$ controls estimation difficulty through the eigenvalue sensitivities: the parameter $\theta_a$ receives signal proportional to $\delta^{a+1}$ from the smallest eigenvalue direction.
For the constant term $\theta_0$ (linear dynamics), the $\delta$-direction contributes $O(\delta^2)$ to the Fisher information; for the highest-degree term $\theta_K$, the contribution degrades to $O(\delta^{2K+2})$.
These scalings hold for both the diagonal entries and the full matrix; the cross-eigenvalue terms share the same $\delta$-dependence.

For linear dynamics (single scalar parameter $\alpha_0$), the Fisher information is exact with no cross-term ambiguity:

\begin{corollary}[Linear Dynamics: Explicit Fisher Information]\label{cor:linear-fisher}
For linear dynamics $\dot{X} = \alpha_0 X$ with snapshots indexed $t=0, \ldots, T$, the Fisher information is:
\begin{equation}
\mathcal{I}(\alpha_0) = 4 \sum_{t=0}^T t^2 \sum_{i < j} \frac{P_{ij}(t)}{1 - P_{ij}(t)}
\end{equation}
This scales as $\Theta(T^3 \cdot n^2)$ when edge probabilities remain bounded away from 0 and 1 throughout the trajectory, giving a Cram\'er--Rao lower bound of $\Omega(1/(T^3 n^2))$.
\end{corollary}

\begin{proof}
See \cref{app:deferred-proofs}.
\end{proof}

The cubic dependence on $T$ (rather than linear) reflects the fact that the sensitivity $\partial P_{ij} / \partial \alpha_0 = 2 t P_{ij}$ grows with time: later snapshots are more informative because the perturbation has had longer to propagate.
This is a general feature of dynamics estimation, in contrast to i.i.d.\ settings where information accumulates linearly.

The scaling $\Theta(T^3 \cdot n^2)$ holds only while the trajectory remains in the interior of the probability space.
Linear dynamics produce exponential growth or decay of eigenvalues: $\lambda_\iota(t) = \lambda_\iota(0) e^{2 \alpha_0 t}$.
For $\alpha_0 > 0$, the probabilities $P_{ij}$ grow toward 1, and the Fisher weights $P_{ij}/(1 - P_{ij})$ diverge; but the model becomes invalid as probabilities saturate, so the growth of $\mathcal{I}$ eventually halts.
For $\alpha_0 < 0$, probabilities decay toward 0, the sensitivities $\partial P_{ij} / \partial \alpha_0 = 2 t P_{ij}(t)$ shrink exponentially, and the Fisher information per snapshot saturates.
In both cases, the $T^3$ scaling describes a \emph{short-time regime} before boundary effects dominate; the effective time horizon is $T_{\mathrm{eff}} \sim 1/|\alpha_0|$.

\textbf{The duality.}
The statistical-geometric correspondence is sharpest for general (non-polynomial) symmetric dynamics, where the full Lyapunov structure appears.
The connection 1-form involves factors $1 / (\lambda_\iota + \lambda_\gamma)$ controlling gauge sensitivity (\cref{sec:fiber-bundle}); the vertical bracket norm in sectional curvature contains the same denominators (weighted as in \cref{prop:curvature-criterion}), so weak eigendirections become geometrically delicate when commutator energy lies in those directions and $\lambda_{d-1}, \lambda_d$ are small; and the Fisher information for estimating $M_{\iota\gamma}$ from the Lyapunov equation $\dot{P}_{\iota\gamma} = (\lambda_\iota + \lambda_\gamma) M_{\iota\gamma}$ involves the same factor $(\lambda_\iota + \lambda_\gamma)^2$.

An important qualification: this duality applies to \emph{full operator recovery} (estimating the action of $M$ on all eigenvector directions, including the weak subspace near $\lambda_d$).
For \emph{parsimonious parameterizations} (e.g., a single scalar parameter $\alpha_0$ shared across all eigenvalue directions), estimation can proceed primarily from the dominant modes: the contribution of $\lambda_1$ alone may suffice to identify $\alpha_0$, bypassing the ill-conditioned $\lambda_d$ direction entirely.
The geometric-statistical coupling bites when the parameter of interest specifically governs the weak subspace, or when the full operator $M$ must be recovered.

With this qualification, the correspondence remains tight: for any component of the dynamics that depends on the spectral gap, the same $\delta$ controls both curvature and Fisher information.
Networks near rank-deficiency ($\delta \to 0$, with $\lambda_{d-1}$ also small) are simultaneously harder to align (curvature $\sim 1/(\lambda_{d-1} + \delta)$), harder to interpolate (injectivity radius $\sim \sqrt{\delta}$), and harder to estimate statistically (Fisher information for the weakest direction $\sim \delta^2$).

For non-polynomial dynamics (e.g., Laplacian), sensitivity propagation is more complex because the eigenvectors of $P$ also evolve.
The additional eigenvector sensitivity contributes positively to the Fisher information (there are more directions in which observations carry signal), but it also introduces holonomy: the extra information can be entangled with gauge degrees of freedom in regions where \cref{prop:laplacian-holonomy} applies.
A complete minimax theory for dynamics estimation in the presence of holonomy remains an important open problem.

\begin{remark}
\textbf{Relationship to ASE estimation theory.}
\Cref{prop:cramer-rao} is derived from the Bernoulli likelihood of the observed adjacency matrices and applies to \emph{any} estimator of $\theta$, not to spectral methods specifically.
The bound is therefore conceptually distinct from the CLT for adjacency spectral embedding~\citep{athreya2016limit,athreya2017statistical}, the two-to-infinity perturbation theory~\citep{cape2019two}, and the efficiency results for spectral estimators of static network parameters~\citep{tang2022efficient,xie2021efficient}.
Those results characterize the accuracy of estimating \emph{latent positions} $X$ or \emph{block probability matrices} $B$ from a \emph{single} graph (or a fixed collection of graphs sharing the same $P$).
ASE is asymptotically unbiased with per-vertex error $O(1/\sqrt{n})$~\citep{athreya2016limit}, and the spectral estimator of SBM block probabilities achieves asymptotic efficiency~\citep{tang2022efficient}; for individual latent positions, a one-step procedure achieves local efficiency (matching the oracle MLE covariance, up to orthogonal transformation)~\citep{xie2021efficient}.
The minimax rate for latent position estimation under Frobenius loss is $\Theta(d/n)$ per vertex~\citep{xie2020optimal}, with two-to-infinity rates depending on the spectral gap~\citep{agterberg2023minimax}.

By contrast, \cref{prop:cramer-rao} concerns estimation of the \emph{dynamics parameters} $\theta$ from a \emph{time series} of graphs, a setting for which no prior efficiency theory exists.
The estimation target is qualitatively different: not the $O(nd)$ latent position coordinates, but the $O(1)$ parameters governing their temporal evolution.
The eigenvalue-direction decomposition of the Fisher information in \cref{prop:cramer-rao}, with weights $\mathcal{W}_\iota(t)$ involving eigenvector loadings, has no counterpart in the static theory.
Extending the existing minimax framework to parametric temporal models is an open problem that our Fisher information structure could help resolve.

The ``any unbiased estimator'' qualification deserves comment.
The CRB is a bound on the Bernoulli likelihood, which is well-defined for any $n$, so the bound itself does not require asymptotic arguments.
However, whether useful estimators of $\theta$ are unbiased at finite $n$ is a separate question.
Any estimator that first recovers $P$ spectrally and then fits $\theta$ inherits the finite-sample bias of eigendecomposition: the nonlinearity of the spectral map introduces bias, but this is $o(1/\sqrt{n})$~\citep{cape2019two} and does not affect the asymptotic bound.
Near rank-deficiency ($\delta \to 0$), the situation is more severe: the ASE collapses the weak eigenvalue direction, and the resulting bias in $\hat{\lambda}_d$ can dominate the signal $\lambda_d$ itself.
In this regime, the \emph{mean squared error} for parameters that depend on the $\lambda_d$ direction is likely dominated by bias rather than variance, and the CRB (which bounds variance alone) understates the true difficulty.
For finite-sample inference, the biased Cram\'er--Rao variant $\mathrm{Var}(\hat{\theta}) \ge (1 + b'(\theta))^2 / \mathcal{I}(\theta)$ provides the appropriate generalization, but characterizing the bias $b(\theta)$ as a function of $\delta$ remains open.
\end{remark}


\section{The Constructive Problem: Identifiability and Its Limits}\label{sec:constructive}

The obstructions in the previous sections characterize what is and isn't observable in principle.
We address the constructive question: given that observable dynamics exist, can we \emph{recover} them from spectral embeddings?
We establish a theoretical identifiability result showing that dynamics structure can, in principle, resolve gauge ambiguity.
We then discuss why this identifiability does not straightforwardly translate into a practical algorithm.

\subsection{Structure-constrained alignment}

The alignment problem is underdetermined: without additional information, many gauge choices produce plausible trajectories.
A natural idea is to use \emph{inductive bias about dynamics structure} to constrain which trajectories are admissible.

Suppose dynamics belong to a family $\mathcal{F}$ from \cref{sec:dynamics-families}. For instance, polynomial dynamics $\dot{X} = N(P) X$ with $N(P) = \sum_k \alpha_k P^k$.
These families have two properties that suggest they could regularize alignment:
they are \emph{horizontal} (symmetric $N$ ensures no gauge drift), and they have \emph{restricted structure} that random gauge errors would violate.

If the true dynamics belong to $\mathcal{F}$, then correct gauges produce trajectories fittable by some $f \in \mathcal{F}$, while wrong gauges produce trajectories requiring dynamics outside $\mathcal{F}$.
The dynamics family acts as regularization on the gauge choice.

We formulate this as a joint optimization problem:

\begin{definition}[Joint Alignment and Learning Problem]
Given ASE embeddings $\{\hat{X}^{(t)}\}_{t=0}^T$ and a dynamics family $\mathcal{F}$, find gauge corrections $\{Q_t \in O(d)\}$ and $f \in \mathcal{F}$ minimizing:
\begin{equation}
\mathcal{L}(\{Q_t\}, f) = \sum_{t=0}^{T-1} \|\hat{X}^{(t+1)} Q_{t+1} - \hat{X}^{(t)} Q_t - \delta t \, f(\hat{X}^{(t)} Q_t)\|_F^2
\end{equation}
\end{definition}

The objective measures how well the aligned discrete trajectory $\{\hat{X}^{(t)} Q_t\}$ is explained by the dynamics $f$: if the gauges are correct and $f$ captures the true dynamics, each step's displacement should match the predicted velocity.

\subsection{The identifiability principle}

The theoretical case for structure-constrained alignment rests on a clean separation between gauge artifacts and horizontal dynamics.

\begin{theorem}[Gauge Velocity Contamination]\label{thm:gauge-contamination}
Let $X(t)$ follow true dynamics $\dot{X} = NX$ with $N = N^\top$.
Let $S(t) \in O(d)$ be a time-varying (mis)gauge and define the gauged trajectory $\tilde{X}(t) = X(t) S(t)$.
Then
\begin{equation}
\dot{\tilde{X}} = N \tilde{X} + \tilde{X} \Omega
\end{equation}
where $\Omega(t) = S(t)^\top \dot{S}(t) \in \so(d)$ is skew-symmetric.

Moreover, if $\tilde{X}(t)$ has full column rank for all $t$ in an interval, then the following are equivalent on that interval:
\begin{enumerate}
\item there exists a symmetric matrix function $\tilde{N}(t) = \tilde{N}(t)^\top$ such that $\dot{\tilde{X}} = \tilde{N}(t) \tilde{X}$,
\item $\Omega(t) = 0$ (i.e., $S(t)$ is constant in time).
\end{enumerate}
\end{theorem}

\begin{proof}
By the product rule,
$\dot{\tilde{X}} = \dot{X} S + X \dot{S} = NXS + X\dot{S}$.
Using $X = \tilde{X} S^\top$ (since $S \in O(d)$), this becomes
$\dot{\tilde{X}} = N \tilde{X} + \tilde{X} (S^\top \dot{S}) = N \tilde{X} + \tilde{X} \Omega$,
where $\Omega = S^\top \dot{S} \in \so(d)$ because differentiating $S^\top S = I$ gives $\dot{S}^\top S + S^\top \dot{S} = 0$.

Now suppose $\dot{\tilde{X}} = \tilde{N} \tilde{X}$ for some symmetric (possibly time-varying) $\tilde{N}$ and that $\tilde{X}$ has full column rank, so $G = \tilde{X}^\top \tilde{X}$ is positive definite.
Left-multiplying by $\tilde{X}^\top$ gives
$\tilde{X}^\top \dot{\tilde{X}} = \tilde{X}^\top \tilde{N} \tilde{X}$,
and the right-hand side is symmetric because $\tilde{N}$ is symmetric.

On the other hand, from the decomposition above,
$\tilde{X}^\top \dot{\tilde{X}} = \tilde{X}^\top N \tilde{X} + G \Omega$.
The term $\tilde{X}^\top N \tilde{X}$ is symmetric, hence $G \Omega$ must be symmetric as well.
But if $G$ is symmetric positive definite and $\Omega$ is skew-symmetric, then $G \Omega$ is symmetric iff $\Omega = 0$ (equivalently, $G \Omega + \Omega G = 0$ implies $\Omega = 0$).
Therefore $\Omega = 0$, which means $S$ is constant in time.
The converse direction is immediate: if $\Omega = 0$ then $\dot{\tilde{X}} = N \tilde{X}$ is already of the required form with symmetric generator $N$.
\end{proof}

The mechanism is elegant: random ASE gauge errors $R^{(t)}$ introduce skew-symmetric contamination that \emph{cannot be absorbed} by symmetric dynamics.
Requiring symmetric dynamics implicitly selects gauges where the contamination vanishes.

The reader may notice an affinity with \cref{prop:horizontal}, which states that $\dot{X}$ is horizontal if and only if $X^\top \dot{X}$ is symmetric.
The two results use the same algebraic fact (a positive definite matrix times a nonzero skew-symmetric matrix is never symmetric) but answer different questions.
\Cref{prop:horizontal} characterizes directions at a single point: given a velocity $\dot{X}$, does it have a gauge component?
\Cref{thm:gauge-contamination} characterizes trajectories over time: given an entire trajectory $\tilde{X}(t)$ observed in the wrong gauge, can any symmetric dynamics explain it?
The former is a pointwise decomposition; the latter is an identifiability statement about the dynamics-gauge coupling along a path.

For specific families, this yields concrete algorithms.
For linear symmetric dynamics, one obtains an alternating scheme: fix gauges and solve a least-squares problem with symmetry constraint for the dynamics matrix $M$; fix $M$ and solve orthogonal Procrustes for each gauge.
For polynomial dynamics $\dot{X} = (\sum_k \alpha_k P^k) X$, the gauge invariance of $P = \hat{X} \hat{X}^\top$ simplifies the dynamics step to linear regression in the $\alpha_k$ coefficients.
Both have closed-form updates per step.

\subsection{The downstream pipeline: from trajectories to equations}

Suppose, optimistically, that the trajectory recovery problem were solved, that is, some method (structure-constrained alignment, a future $P$-level estimator, or domain-specific prior information) produced gauge-consistent estimates $\tilde{X}^{(t)}$ of the latent positions, up to noise.

\begin{remark}
\textbf{Sampling-averaging heuristic.}
When multiple independent network samples $A_1^{(t)}, \ldots, A_m^{(t)}$ are available at each time $t$, averaging
$\bar{A}^{(t)} = \frac{1}{m} \sum_i A_i^{(t)}$
reduces per-entry variance by a factor of $m$.
In standard dense perturbative regimes, one therefore expects roughly a $\sqrt{m}$ improvement in per-vertex embedding error.
We use this as a practical scaling heuristic, not as a proved rate in this manuscript.
\end{remark}

Given such a recovered trajectory, the remaining steps are well-established.
\emph{Universal Differential Equations} (UDEs)~\citep{rackauckas2020universal} provide a framework for learning dynamics that combine known mechanistic structure with neural network flexibility:
\begin{equation}
\dot{X} = g(f_{\text{known}}(X, \phi), f_{\text{NN}}(X, \theta))
\end{equation}
For RDPG dynamics, the known structure comes from the families in \cref{sec:dynamics-families}. For example, the polynomial architecture $\dot{X} = (\sum_k \alpha_k(\theta, X) P^k) X$ ensures horizontality by construction while allowing the coefficients to be learned.
Gradients flow through the ODE solver via adjoint sensitivity methods, and the gauge-consistent architecture constrains the hypothesis space.

Once a UDE is trained, \emph{symbolic regression}~\citep{symbolicregression} can extract interpretable closed-form equations.
By sampling state-velocity pairs $(X, f_\theta(X))$ from the trained model and searching for symbolic expressions of the form $N(P) X$, one can recover explicit dynamics equations completing the path from observed adjacency matrices to differential equations.
A subtlety: symbolic regression on latent variables $X$ is gauge-dependent.
Unless the UDE output is pre-aligned to a canonical frame (e.g., the eigenbasis of $P$), the regression may fail to find sparse coefficients because $X$ in the UDE gauge is a rotated version of the ``natural'' coordinates.
Working with \emph{gauge-invariant features} (eigenvalues of $P$, or the $P$-dynamics $\dot{P} = NP + PN$ directly) avoids this issue, and is natural for polynomial families where the dynamics separate by eigenvalue.

The bottleneck is obtaining $\tilde{X}^{(t)}$.
The UDE and symbolic regression machinery is mature; the trajectory recovery problem is not.

\subsection{Why the constructive problem remains open}

Despite the clean identifiability theory, translating \cref{thm:gauge-contamination} into a reliable recovery method faces three interacting difficulties.

\textbf{Finite-sample bias.}
Each spectral embedding $\hat{X}^{(t)}$ estimates the true positions with error $O_p(1/\sqrt{n})$.
When the true dynamics are slow (small $\|X^{(t+1)} - X^{(t)}\|$), the signal-to-noise ratio for alignment degrades: we are trying to distinguish small true displacements from estimation noise of comparable magnitude.
The alternating optimization inherits this: the dynamics step fits a combination of true signal and noise, while the gauge step aligns to a noisy target.
The two sources of error reinforce rather than cancel.

\textbf{Expressiveness of dynamics families.}
The identifiability argument needs $\mathcal{F}$ to be restrictive enough that wrong gauges force you outside the family.
But in finite samples and discrete time, surprisingly flexible behavior can sneak in.

One mechanism is numerical rather than philosophical: even for polynomial families, the regression problem for $(\alpha_0, \ldots, \alpha_K)$ can become ill-conditioned.
The matrices $I, P, P^2, \ldots$ increasingly align with the top eigenspace of $P$ (power-iteration style), so the design becomes Vandermonde-like in the eigenvalues and can be badly conditioned for $K \ge 2$.
In that regime, modest gauge-induced artifacts (or plain ASE noise) can be fitted by large, canceling coefficients in higher-order terms, making it hard to distinguish ``true dynamics'' from ``dynamics that explain the mis-gauged data''.

A second mechanism is geometric: for small $d$, the gauge degrees of freedom are low-dimensional (e.g., $\so(2)$ is one-dimensional), so a slowly varying mis-gauge can produce a coherent-looking drift that is not obviously nonsense at the discrete-time resolution.
The moral is the same: making $\mathcal{F}$ more restrictive risks excluding the truth; making it more expressive risks letting gauge artifacts (and noise) slip through.

\textbf{Holonomy and global consistency.}
Even if local alignment succeeds (each consecutive pair of frames is well-aligned), holonomy (\cref{sec:fiber-bundle}) means gauge drift can accumulate over long trajectories.
For dynamics whose \emph{base-space} path $P(t)$ forms (or nearly forms) loops in $\mathcal{B}$, a globally consistent choice of gauge along the entire loop may be impossible: the horizontal lift of a closed curve need not close.

This is not something you fix with more data; it is a geometric obstruction of the model class.
What more data \emph{can} do is reduce the statistical wobble on top of it. That is important, but it is a different issue.

\textbf{The fundamental tradeoff.}
The core difficulty is a tension between two requirements.
The dynamics family must be restrictive enough to reject gauge-contaminated trajectories (identifiability), but expressive enough to capture the true dynamics (model adequacy).
For the problem to be well-posed, the ``gap'' between what $\mathcal{F}$ can fit and what gauge contamination produces must exceed the noise level.
Characterizing when this gap is sufficient (as a function of $n$, $d$, $T$, the spectral gap of $P$, and the dynamics family) is the central open problem.

\begin{remark}
The difficulty is not that the identifiability principle fails theoretically.
\Cref{thm:gauge-contamination} is sharp: in the continuous-time, infinite-data limit, symmetric dynamics \emph{do} identify gauges.
The difficulty is that the finite-sample, discrete-time setting introduces errors that interact with the gauge-dynamics coupling in ways that the asymptotic theory does not capture.
Progress likely requires either stronger structural assumptions (e.g., working directly on $P$-dynamics to bypass the gauge problem entirely) or fundamentally different estimation strategies (e.g., information-theoretic approaches that characterize the minimax rate for dynamics recovery).
\end{remark}

\subsection{Anchor-based alignment: a tractable special case}\label{sec:anchor-alignment}

The difficulties above are genuine in the general case.
However, there is a natural special case in which the gauge problem admits a clean solution: when a subset of nodes is known (or assumed) to be stationary.

\textbf{The anchor principle.}
Suppose a subset $S \subset \{1, \ldots, n\}$ of nodes has $\dot{x}_i = 0$ for $i \in S$. These are ``anchor'' points whose latent positions do not move.
At each time $t$, ASE produces $\hat{X}^{(t)} = X^{(t)} R^{(t)} + E^{(t)}$ with a gauge factor $R^{(t)} \in O(d)$ (coming from the eigendecomposition choice) and estimation noise $E^{(t)}$.
For anchor nodes $i \in S$, we have $x_i^{(t)} = x_i^{(0)}$ for all $t$, so the anchor rows satisfy
$\hat{X}_S^{(t)} = X_S R^{(t)} + E_S^{(t)}$
where $X_S \in \RR^{|S| \times d}$ is constant.

Align the anchor rows to a reference frame (say $t = 0$) via Procrustes:
\begin{equation}
\hat{Q}^{(t)} = \arg \min_{Q \in O(d)} \|\hat{X}_S^{(t)} Q - \hat{X}_S^{(0)}\|_F.
\end{equation}
In the noiseless limit this recovers the relative gauge $(R^{(t)})^{-1} R^{(0)}$ directly, without using any dynamics model and without ever propagating a gauge sequentially.
Applying $\hat{Q}^{(t)}$ to the full embedding $\hat{X}^{(t)}$ then aligns all nodes to the $t=0$ gauge.

\textbf{Why this works.}
The anchor nodes provide a fixed reference frame that ``pins'' the gauge at each time step.
The alignment is \emph{global} (all frames aligned to $t = 0$, not sequentially) so there is no error accumulation.
The method requires no knowledge of the dynamics family $\mathcal{F}$: it works for polynomial, Laplacian, or any other dynamics, as long as the anchors are truly stationary.
Holonomy is irrelevant because we never attempt to propagate a gauge along the trajectory; each frame is independently aligned to the fixed reference.

\textbf{Conditions and limitations.}
The anchor principle requires:
\begin{enumerate}
\item \emph{Well-conditioned anchors:} the anchor position matrix $X_S \in \RR^{|S| \times d}$ must have full column rank for the Procrustes problem to be determined, and should be well-conditioned ($\sigma_d(X_S)$ not too small) for robustness to noise. Counting anchors alone is insufficient: $|S| \gg d$ anchors that are nearly collinear or clustered around a single point leave the rotation poorly constrained in orthogonal directions. This is the same stability story as in the classical orthogonal Procrustes problem: the rotation estimate becomes sensitive when the reference configuration is nearly rank-deficient. (We do not reproduce a full perturbation bound here; the key point for our purposes is the conditioning dependence. See, e.g., standard treatments of Procrustes perturbation theory, or interpret it directly through the SVD formula for Procrustes and Weyl/Wedin-type eigen/singular vector perturbation bounds.)
\item \emph{Known anchor identity:} we must know which nodes are stationary. In practice this could come from domain knowledge (e.g., established species in an ecological network, institutional nodes in a social network) or from a preliminary analysis identifying nodes with low temporal variance.
\item \emph{Approximately stationary anchors:} if anchor nodes drift slowly (with velocity $\|\dot{x}_i\| = \epsilon$ for $i \in S$), the alignment incurs a systematic bias that grows with the time horizon. A crude but useful rule of thumb is that anchor drift should remain small compared to the embedding noise accumulated over the window of interest; otherwise, the ``reference frame'' itself becomes moving.
\end{enumerate}

\begin{remark}
The third condition explains a phenomenon we observed in preliminary experiments: networks with a large block of slowly-moving nodes could be aligned successfully by naive Procrustes, while networks where all nodes moved at comparable rates could not.
The slow nodes were acting as \emph{de facto} anchors, stabilizing the gauge without our explicitly recognizing it.
\end{remark}

\textbf{When anchors are realistic.}
Several application domains naturally feature nodes with heterogeneous dynamics rates.
In ecological food webs, basal species (primary producers) often have stable trophic positions while higher-level consumers undergo rapid changes.
In social networks, institutional actors (organizations, permanent positions) may persist while individuals fluctuate.
In neural connectomes, structural hub regions may be stable on timescales over which peripheral connections rewire.
More generally, any system with a separation of timescales is a natural candidate: it has a slowly evolving ``backbone'' and a rapidly evolving periphery.

The anchor approach does not solve the general gauge problem.
It replaces a hard geometric-statistical problem with a domain knowledge requirement: identifying stationary nodes.
But when such knowledge is available, it provides a clean and computationally trivial path to gauge-consistent trajectories, enabling the downstream UDE pipeline of \cref{sec:constructive}.

\subsection{Numerical illustration}\label{sec:numerics}

We illustrate the theory using two controlled experiments on synthetic RDPG data.
The first demonstrates anchor-based alignment on gauge-equivariant polynomial dynamics, where alignment quality can be assessed independently of dynamics recovery.
The second demonstrates the full downstream pipeline: UDE training and symbolic regression; to do this more visibly, we move to a non-gauge-equivariant dynamics, where alignment quality directly impacts dynamics recovery.

\subsubsection{Experiment 1: Anchor-based alignment}\label{sec:numerics-anchor}

\textbf{Setup.}
We generate an RDPG with $n = 200$ nodes in $d = 2$ latent dimensions.
The initial positions $X(0)$ are drawn uniformly from $B_+^2 = \{x \in \RR^2 : x_1, x_2 \ge 0, \|x\| \le 1\}$, with a designated anchor set $S$ of size $n_a$.
The non-anchor nodes evolve under polynomial dynamics $\dot{X} = (\alpha_0 I + \alpha_1 P) X$ with $(\alpha_0, \alpha_1) = (-0.3, 0.003)$, while anchor nodes remain fixed: $\dot{x}_i = 0$ for $i \in S$.
We integrate the ODE to produce a trajectory $X(t)$ at $T$ equally-spaced times, generate $m = 3$ independent adjacency matrices per time step $A_k^{(t)} \sim \text{Bernoulli}(X(t) X(t)^\top)$ (and average them to reduce noise), compute ASE embeddings $\hat{X}^{(t)}$, and compare anchor-based alignment to sequential Procrustes.

Unless otherwise stated, we use $T = 50$ steps with $\delta t = 0.05$ (total time $2.45$), and report means and standard deviations over 20 Monte Carlo repetitions.

\textbf{Metrics.}
We measure alignment quality by the mean Procrustes error:
\begin{equation}
\mathrm{err}(t) = \frac{1}{n} \|\hat{X}^{(t)} \hat{Q}^{(t)} - X^{(t)} Q^*\|_F
\end{equation}
where $Q^*$ is the best global alignment to the true trajectory (accounting for the residual gauge at $t = 0$).
We sweep over four experimental conditions: anchor count $n_a$, trajectory length $T$, anchor drift rate $\epsilon$, and initial embedding norm scale (controlling signal-to-noise ratio).

\begin{figure}[ht]
\centering
\includegraphics[width=0.95\textwidth]{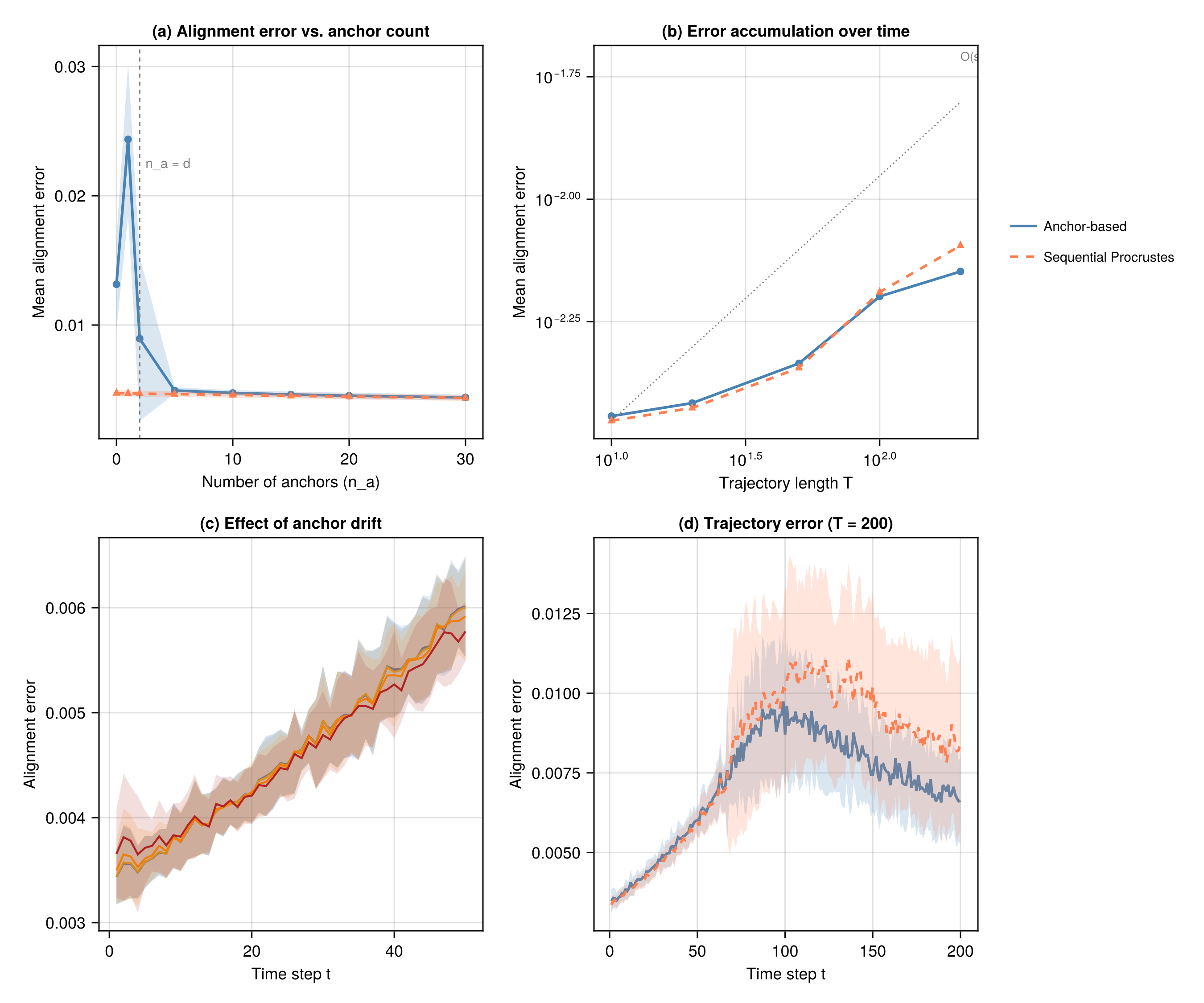}
\caption{
Anchor-based alignment experiment ($n = 200$, $d = 2$, polynomial dynamics $\dot{X} = (\alpha_0 I + \alpha_1 P) X$ with $m = 3$ Bernoulli samples per time step).
\textbf{(a)} Alignment error vs.\ number of anchor nodes: below $n_a = d = 2$ (dashed vertical line) Procrustes is underdetermined and alignment fails; with sufficient anchors, error stabilizes at the ASE noise floor.
\textbf{(b)} Error accumulation over trajectory length: anchor-based alignment (blue) remains bounded while sequential Procrustes (coral, dashed) grows with $T$, consistent with $O(\sqrt{T})$ drift accumulation.
\textbf{(c)} Effect of anchor drift rate $\epsilon$: with larger $\epsilon$, systematic bias from drifting anchors becomes visible.
\textbf{(d)} Per-timestep alignment error for the $T = 200$ trajectory: the anchor-based error (blue) remains approximately flat, while sequential Procrustes (coral) accumulates drift, with the gap widening toward later time steps.
Shaded bands show $\pm 1$ standard deviation across 20 Monte Carlo repetitions.
}
\label{fig:anchor-main}
\end{figure}

\begin{figure}[ht]
\centering
\includegraphics[width=0.95\textwidth]{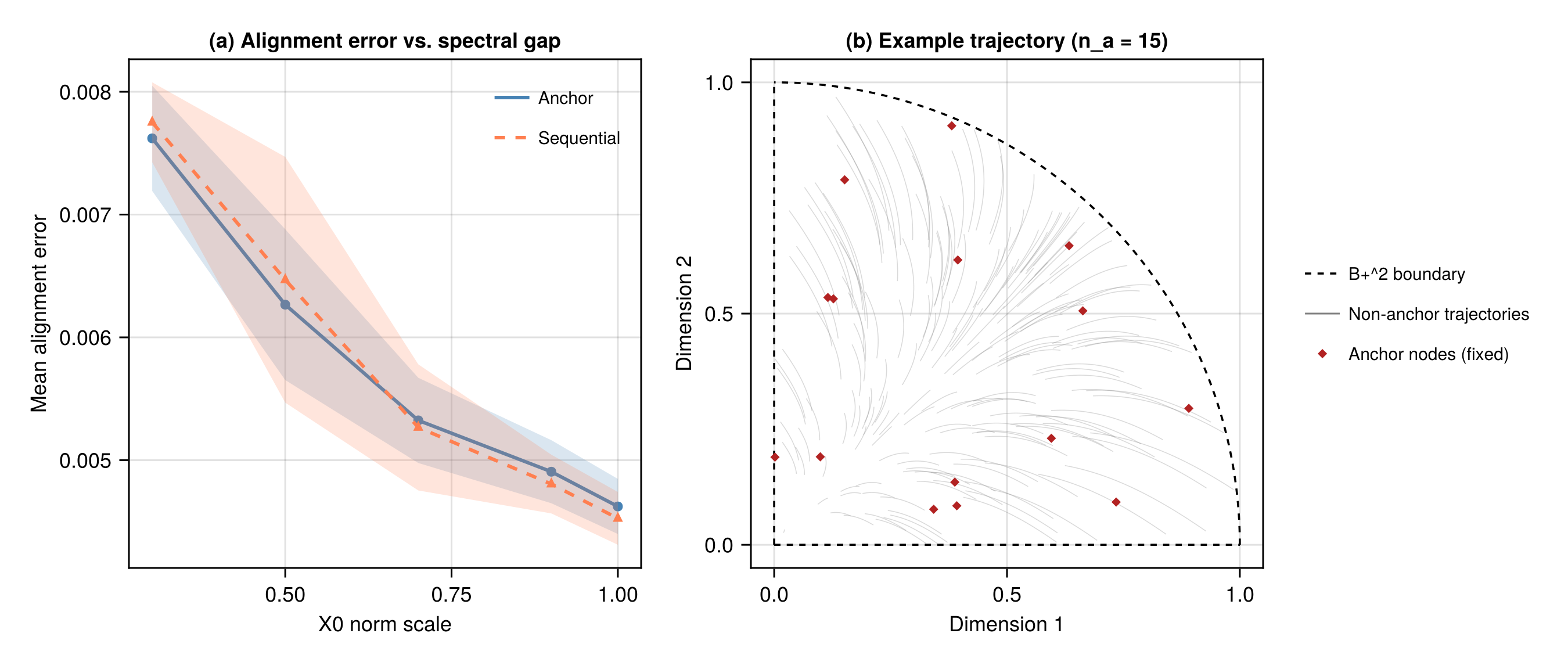}
\caption{
\textbf{(a)} Alignment error as a function of the initial embedding norm scale (proxy for signal strength): smaller norms yield weaker signals and less informative Bernoulli observations, increasing ASE noise and degrading alignment for both methods. Note that uniform scaling preserves the condition number $\sigma_1 / \sigma_2$ (constant at $\approx 2.2$); this experiment varies signal \emph{magnitude} rather than the spectral gap ratio.
\textbf{(b)} Phase portrait with $n_a = 15$ anchor nodes (red diamonds, stationary) and non-anchor nodes (gray trajectories evolving under polynomial dynamics) in $B_+^2$.
}
\label{fig:anchor-spectral}
\end{figure}

\textbf{Results.}
The experiment matches the basic theoretical story, and also shows where the ``easy'' intuitions break.

\begin{itemize}
\item \emph{Anchor count.} With too few anchors the Procrustes problem is underdetermined: for $d=2$, $n_a=1$ is a worst case and alignment fails, while $n_a=2$ is a knife-edge case with high variance (sometimes it works, sometimes it doesn't, depending on anchor geometry). Once $n_a$ is moderately larger (empirically $n_a \ge 5$ here), the anchor-based error stabilizes near the ASE noise floor (about $0.005$ in our setting), consistent with the idea that the anchors now robustly span the latent plane.

\item \emph{Trajectory length and drift accumulation.} For short trajectories ($T \le 50$), sequential Procrustes is marginally better (about $2\%$ in our sweep), because it uses all $n=200$ nodes at each step while anchor-based Procrustes uses only the $n_a = 15$ anchors. But sequential alignment accumulates drift: the crossover occurs around $T \approx 100$, and by $T = 200$ sequential Procrustes is about $13\%$ worse (ratio $\approx 1.13$), consistent with the expected $O(\sqrt{T})$ growth of accumulated rotation error.

\item \emph{Drifting anchors.} Over $T=50$ steps (total time $2.45$), modest anchor drift rates have little effect on the mean alignment error (because the metric averages over all 200 nodes and only 15 are anchors). The effect is more visible in downstream variability: at $\epsilon = 0.1$ the variability of recovered parameters increases noticeably, consistent with systematic bias entering through a slowly moving reference frame.
\end{itemize}

\textbf{Dynamics recovery is (mostly) gauge-free here.}
A useful sanity check is that, for polynomial dynamics, the coefficients $(\alpha_0, \alpha_1)$ are estimable directly from the gauge-invariant trajectory $P(t) = X(t) X(t)^\top$.
In our sweep, the standard errors of $\hat{\alpha}_0$ are essentially unchanged across anchor counts (around $0.008$), confirming that this particular parameter recovery problem does not fundamentally depend on aligning $X$.

However, the \emph{point estimates} do shift with the anchor count (e.g., $\hat{\alpha}_0$ moves from about $-0.282$ at $n_a=0$ to about $-0.250$ at $n_a=15$), because freezing a large subset of nodes changes the actual $P(t)$ trajectory. This is a data-generating change, not a gauge artifact.
This motivates the next question: when the dynamics depend on $X$-space coordinates (i.e., are not gauge-equivariant), does alignment quality directly control dynamics recovery?
Experiment 2 is designed to make that dependence unavoidable.

\subsubsection{Experiment 2: UDE pipeline with non-gauge-equivariant dynamics}\label{sec:numerics-ude}

\textbf{Motivation.}
The polynomial dynamics in Experiment 1 are gauge-equivariant: $P = XX^\top$ is rotation-invariant, so the dynamics parameters live in $P$-space and can be recovered without alignment.
To exercise the full trajectory-recovery $\to$ UDE pipeline, we now design dynamics that depend on $X$-space coordinates directly, so that alignment quality becomes a genuine bottleneck.

\textbf{Setup.}
We generate an RDPG with $n = 200$ nodes in $d = 3$ latent dimensions, with $n_a = 100$ anchor nodes organized into $K_c = 3$ communities near the vertices of $B_+^3$ (at positions $(0.7, 0.2, 0.2)$, $(0.2, 0.7, 0.2)$, $(0.2, 0.2, 0.7)$ plus Gaussian noise).
The 100 non-anchor nodes evolve under damped spiral dynamics around their community centroids $\mu_k$:
\begin{equation}
\dot{x}_i = (-\gamma + \beta \|x_i - \mu_k\|^2)(x_i - \mu_k) + \omega J(x_i - \mu_k)
\end{equation}
where $J = \frac{1}{\sqrt{3}} \bigl(\begin{smallmatrix} 0 & -1 & 1 \\ 1 & 0 & -1 \\ -1 & 1 & 0 \end{smallmatrix}\bigr)$ is the rotation generator around the $(1,1,1)/\sqrt{3}$ axis, and $(\gamma, \beta, \omega) = (0.3, -0.5, 1.0)$.
These dynamics are \emph{not} gauge-equivariant: the centroids $\mu_k$ and rotation axis live in $X$-space coordinates, so a misaligned trajectory $\tilde{X}(t)$ produces different offsets $\tilde{x}_i - \mu_k \neq x_i - \mu_k$, corrupting the dynamics structure.

We observe $m = 10$ independent adjacency matrices per time step over $T = 50$ steps at $\delta t = 0.1$, compute ASE embeddings, and apply three alignment conditions: anchor-based Procrustes, sequential Procrustes, and no alignment (raw ASE with random per-frame rotations).

\textbf{UDE architecture.}
Following the framework of \cref{sec:constructive}, we decompose the dynamics into known and unknown components:
\begin{equation}
\dot{x}_i = \underbrace{-\hat{\gamma} (x_i - \mu_k)}_{f_{\text{known}}} + \underbrace{f_\theta(x_i - \mu_k)}_{f_{\text{NN}}}
\end{equation}
where $\hat{\gamma}$ is a learnable scalar parameter and $f_\theta: \RR^3 \to \RR^3$ is a small neural network (architecture: $3 \to 16 \to 16 \to 3$ with $\tanh$ activations, $\approx 390$ parameters total).
The known part encodes the structural assumption that non-anchor nodes are attracted toward their community centroids; the unknown part captures whatever additional dynamics exist.
L2 regularization on the network weights ($\lambda = 10^{-3}$) discourages the network from absorbing the linear damping term, improving the identifiability of the known-unknown decomposition.
The UDE is trained by solving the neural ODE forward from the initial condition and minimizing the trajectory MSE via adjoint sensitivity methods.

\textbf{Primary metric: total dynamics MSE (the main story).}
Because the NN can absorb part of the linear damping term, the additive decomposition
$f_{\text{learned}}(\delta) = -\hat{\gamma} \delta + f_\theta(\delta)$
has an inherent identifiability caveat: different $(\hat{\gamma}, f_\theta)$ pairs can yield essentially the same total vector field.
So we evaluate what we actually care about, namely the \emph{total learned dynamics}, against the true dynamics
$f_{\text{true}}(\delta) = (-\gamma + \beta \|\delta\|^2) \delta + \omega J \delta$
on a held-out cloud of random test inputs $\delta \in [-0.3, 0.3]^3$ (2000 samples in our implementation).
This metric is intentionally invariant to how the model splits responsibility between $\hat{\gamma}$ and the NN.

\begin{figure}[ht]
\centering
\includegraphics[width=0.95\textwidth]{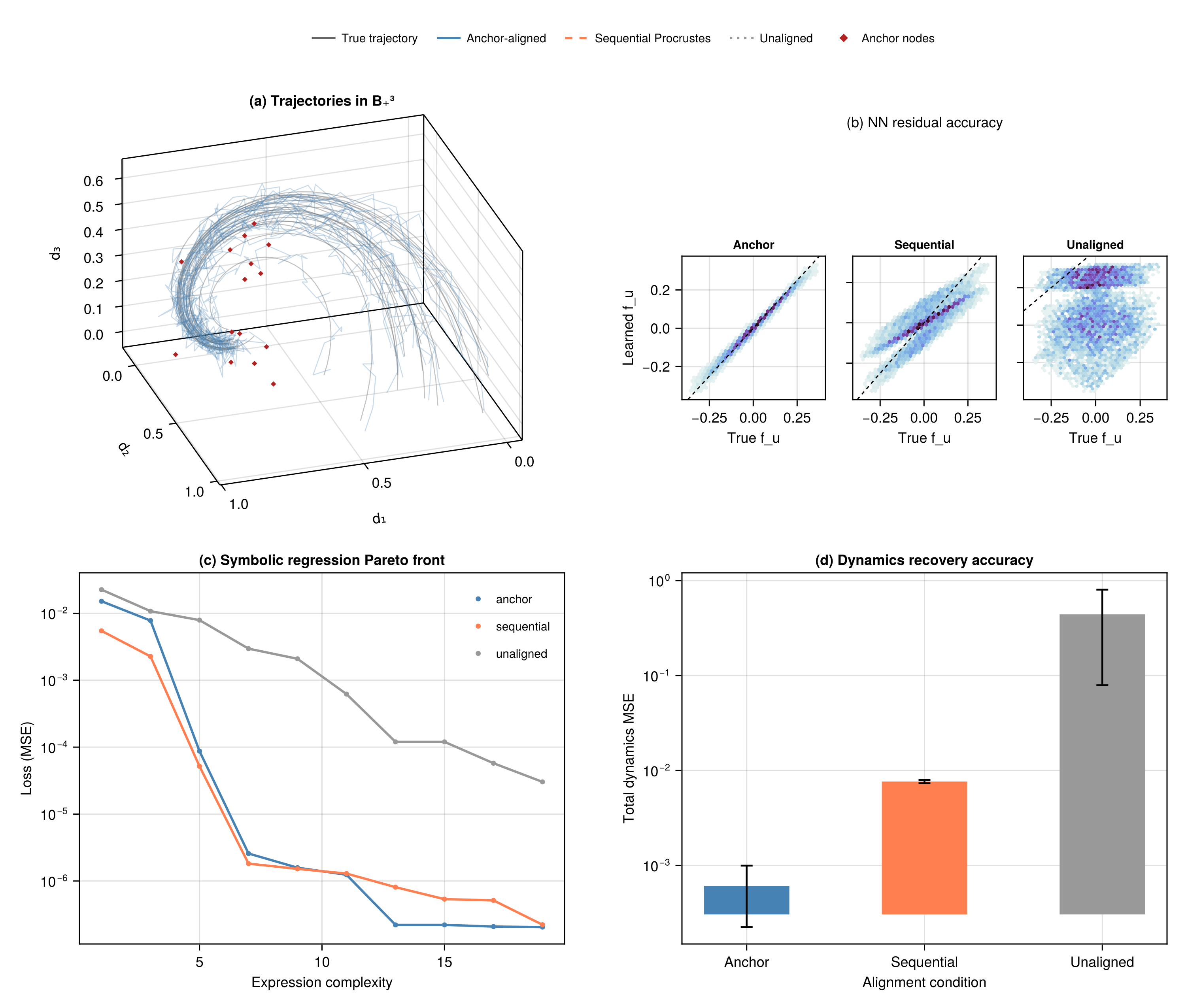}
\caption{
UDE pipeline experiment ($n = 200$, $d = 3$, damped spiral dynamics with rotation around $(1,1,1)/\sqrt{3}$, $m = 10$ Bernoulli samples per frame).
\textbf{(a)} True trajectories (gray) and anchor-aligned ASE (blue) in $B_+^3$; anchor nodes shown as red diamonds.
\textbf{(b)} Learned NN residual $f_\theta(\delta)$ vs.\ true residual $f_u(\delta)$: anchor-aligned data (left) produces tight agreement along the diagonal; sequential Procrustes (center) and unaligned data (right) degrade progressively.
\textbf{(c)} Symbolic regression Pareto fronts (loss vs.\ expression complexity): anchor-aligned data achieves $1$--$2$ orders of magnitude lower loss at each complexity level.
\textbf{(d)} Total dynamics MSE (log scale) by alignment condition: anchor alignment achieves MSE $\approx 6 \times 10^{-4}$, sequential Procrustes $\approx 8 \times 10^{-3}$ ($13\times$ worse), and unaligned $\approx 0.44$ ($\approx 700\times$ worse).
Error bars show $\pm 1$ standard deviation across 5 repetitions.
}
\label{fig:ude-pipeline}
\end{figure}

\textbf{Results.}
The total dynamics MSE separates the three alignment conditions extremely clearly (\cref{fig:ude-pipeline}d), across 5 Monte Carlo repetitions:

\begin{itemize}
\item \emph{Anchor-aligned:} mean MSE $\approx 6 \times 10^{-4}$ with standard deviation $\approx 4 \times 10^{-4}$, i.e., excellent recovery of the full vector field.
\item \emph{Sequential Procrustes:} mean MSE $\approx 7.6 \times 10^{-3}$ with standard deviation $\approx 3 \times 10^{-4}$, about $13\times$ worse.
\item \emph{Unaligned:} mean MSE $\approx 0.44$ with standard deviation $\approx 0.35$, about $700\times$ worse and highly variable.
\end{itemize}

The NN-residual MSE shows the same ordering even more starkly: anchor alignment yields NN residual MSE around $9 \times 10^{-4}$, sequential around $7.3 \times 10^{-3}$, and unaligned around $0.44$.
This is the mechanism in plain sight: when frames are coherently aligned, the NN learns the intended nonlinear+rotational residual; when frames are incoherently rotated, it mostly learns noise.

\textbf{Symbolic regression.}
Symbolic regression on the trained network's input-output pairs (\cref{fig:ude-pipeline}c) achieves Pareto-optimal expressions with losses $1$--$2$ orders of magnitude lower under anchor alignment than under the other conditions at each complexity level.
The anchor-aligned expressions contain the expected structural terms (products of coordinates and cross-terms consistent with the rotation generator $J$).
As noted in \cref{sec:constructive}, symbolic regression requires working with gauge-invariant features or a canonically aligned frame; the anchor alignment provides the latter.

\textbf{The identifiability caveat (and why it doesn't undercut the result).}
The learned damping parameter $\hat{\gamma}$ is not particularly stable across repetitions (anchor-aligned mean $\approx 0.255$ with std $\approx 0.081$ versus true $\gamma=0.3$), because the NN can still absorb part of the linear term even with L2 regularization.
This is a real limitation of the additive UDE decomposition, not an alignment issue.
The reason it does not undermine the main claim is that the \emph{total dynamics} are nevertheless recovered accurately under anchor alignment, and that is exactly what the total-dynamics MSE measures.

\textbf{Takeaway.}
When dynamics are gauge-equivariant ($\dot{X} = N(P) X$), alignment quality is irrelevant for parameter recovery (Experiment 1).
When dynamics depend on $X$-space coordinates, alignment quality directly controls the fidelity of the learned dynamics (Experiment 2).
Anchor-based alignment provides the gauge-consistent trajectories that the downstream UDE pipeline requires.


\section{Discussion and conclusion}\label{sec:discussion}

We hope to have convinced the reader that treating a Random Dot Product Graph as a dynamical system is appealing, even if there are geometric and statistical reasons that make it genuinely difficult.
The main gain we get by framing the temporal RDPG as a dynamical system is a rigorous theoretical understanding, together with a set of concrete failure modes, that clarifies where the difficulty originates and how they manifest.

We identified three core obstructions.
First, gauge freedom makes an entire subspace of latent motion invisible (\cref{thm:invisible}) and forces any latent-space method to confront alignment.
Second, realizability constrains which probability-matrix dynamics are even possible at fixed latent dimension.
Third, the trajectory-recovery problem shows that in the standard RDPG pipeline---where latent positions are estimated via ASE---naive time differencing predominantly measures gauge jitter rather than dynamics.

The differential geometric setting we adopted and developed (\cref{sec:fiber-bundle}) formalizes gauge freedom as a principal bundle and makes the connection/curvature/holonomy machinery explicit in the RDPG setting.
The connection 1-form formalizes ``gauge velocity,'' and curvature/holonomy explain why local alignment may fail to globalize.
The analysis of concrete dynamics allowed us to identify a robust holonomy contrast (\cref{sec:holonomy-dynamics}): polynomial dynamics act through commuting generators (trivial holonomy along trajectories), whereas Laplacian dynamics satisfy a proved local nontrivial-holonomy criterion and, in $d=2$, a full restricted holonomy $\mathrm{SO}(2)$ consequence; for $d \ge 3$, we give a conditional full-holonomy criterion and keep the generic full-holonomy statement conjectural.
The information-theoretic analysis (\cref{sec:info-theoretic}) highlights a deep statistics-geometry duality: the same spectral quantities that control curvature and injectivity also control Fisher information and ill-conditioning.

The difficulties we exposed do not leave us without hopes.
For example, the identifiability principle (\cref{thm:gauge-contamination}) shows that structure can resolve gauge ambiguity as symmetric dynamics cannot absorb skew-symmetric gauge contamination.
In some scenarios, the alignment problem \emph{can} be solved, and we showed it in a tractable special case with anchor nodes (\cref{sec:anchor-alignment}). When this is possible, the downstream UDE pipeline becomes practically feasible and attractive (as the numerical results illustrate).

Yet, the discrete, finite-sample regime is where the theory meets its limits.
Finite-sample spectral bias and noise interact with alignment; expressive dynamics families can fit artifacts; and holonomy can make global gauge consistency impossible over loops, even with perfect local alignment.
Bridging this gap is an open problem. In particular, it would be valuable to characterize when structure-constrained alignment is stable and when it is fundamentally unstable.

Several avenues seem promising for narrowing the constructive gap, both by bypassing gauge when possible and by quantifying the geometric effects when it cannot be avoided.
\begin{enumerate}
\item \emph{Work in $P$-space when possible.} The $P$-dynamics viewpoint (\cref{sec:p-dynamics}) bypasses gauge in principle. The Euclidean mirror framework of \citet{athreya2024euclidean} demonstrates that a gauge-invariant scalar summary of the $P$-trajectory (the spectral-norm Procrustes distance between network states) can be consistently estimated and is practically useful for change-point detection. However, learning dynamics requires the full matrix $\dot{P}$, not only the scalar inter-time distance. At present, we lack an estimation theory for recovering $\dot{P}$ and inverting the Lyapunov structure from Bernoulli samples. The $1/(\lambda_\iota + \lambda_\gamma)$ amplification suggests that minimax rates will depend sharply on the spectral gap; pinning this down would clarify when $P$-level dynamics inference is practically viable and whether the mirror's estimation guarantees can be extended to the richer differential structure required here.
\item \emph{Quantify holonomy, not only its existence.} Beyond generic nontriviality, one would like trajectory-dependent ``holonomy budgets'': how much gauge drift should be expected over a cycle as a function of curvature and of proximity to rank deficiency?
\item \emph{Achievability of the CRB.} For polynomial dynamics (trivial holonomy, explicit eigenvalue ODEs), a matching upper bound via likelihood-based estimators on the $P$-trajectory appears plausible. For Laplacian dynamics, holonomy suggests that achieving the CRB may require estimators that explicitly account for gauge transport.
\item \emph{Sparsity changes everything.} In the sparse regime $P = \rho_n XX^\top$ with $\rho_n \to 0$, injectivity shrinks, and both curvature and ASE error deteriorate. The combined effect on dynamics recovery remains largely unexplored.
\end{enumerate}

Our results have direct interest in applications. Learning interpretable dynamics from network data has the potential to be useful across neuroscience, ecology, and social systems. Characterizing which dynamics are possible \emph{within} the framework helps inform modelling decisions.
As for any model, the RDPG assumptions are structurally strong, and learned dynamics should be treated as hypotheses to be stress-tested against domain knowledge and model misspecification.


\section*{Acknowledgments}

We are extremely grateful to our student Connor Stirling James Smith, who pioneered this research in his Thesis work and had to deal with the challenges of working with non-gauge-equivariant dynamics before the mathematical settings was developed.

\section*{Code Availability}

Julia code for reproducing the results and data are available at \url{https://github.com/gvdr/RDPG_in_differential_geometry}.


\bibliographystyle{plainnat}
\bibliography{bibliography}


\clearpage
\appendix

\section{Deferred Proofs}\label{app:deferred-proofs}

\textbf{Proof of \cref{prop:measure-zero} (Measure Zero).}
Assume $\mathcal{F}$ is a finite-dimensional family with $p$ parameters (e.g., polynomial dynamics with $p = K+1$ coefficients), and restrict attention to a parameter domain on which solutions exist on $[0,T]$ and depend smoothly on parameters and initial conditions (standard under local Lipschitz conditions in $X$ and smoothness in parameters; see Lee, \emph{Introduction to Smooth Manifolds}, 2nd ed., 2013, book-level reference).
For each $(f, X_0) \in \mathcal{F} \times \RR^{n \times d}$, the Picard--Lindel\"of theorem yields a unique solution $X(\cdot; f, X_0) \in C^k([0,T], \RR^{n \times d})$ (same reference).
The solution map $\Phi: (f, X_0) \mapsto X(\cdot; f, X_0)$ is smooth, so $\mathcal{M}_\mathcal{F} = \im(\Phi)$ is the image of a smooth map from an $m$-dimensional manifold (with $m = p + nd$) into the path space $\mathcal{X} = C^k([0,T], \RR^{n \times d})$.

Choose $m + 1$ continuous linear functionals $\ell_1, \ldots, \ell_{m+1}$ on $\mathcal{X}$ (for example, evaluations of specific coordinates at distinct times) and define $L = (\ell_1, \ldots, \ell_{m+1}): \mathcal{X} \to \RR^{m+1}$.
Choose these functionals so that the pushforward $L_* \mu$ is non-degenerate (equivalently, has full-rank covariance), hence absolutely continuous with respect to Lebesgue measure on $\RR^{m+1}$.
Then $L \circ \Phi: \RR^m \to \RR^{m+1}$ is a smooth map between finite-dimensional spaces with $m < m+1$.
A standard consequence of Sard's theorem (equivalently, the area formula) is that the image of such a map has Lebesgue measure zero in $\RR^{m+1}$ (book-level reference: Milnor, \emph{Topology from the Differentiable Viewpoint}, 1997 ed., Sard theorem section).
Thus $L(\mathcal{M}_\mathcal{F}) \subseteq \im(L \circ \Phi)$ also has Lebesgue measure zero, and therefore
$\mu(\mathcal{M}_\mathcal{F}) \le (L_* \mu)(L(\mathcal{M}_\mathcal{F})) = 0$. \qed

\medskip

\textbf{Proof of \cref{cor:linear-fisher} (Linear Dynamics: Explicit Fisher Information).}
For linear dynamics $\dot{X} = \alpha_0 X$, the solution is $X(t) = e^{\alpha_0 t} X(0)$, giving $P(t) = e^{2 \alpha_0 t} P(0)$.
The sensitivity is $\partial P_{ij}(t) / \partial \alpha_0 = 2t e^{2 \alpha_0 t} P_{ij}(0) = 2t P_{ij}(t)$.
Since there is a single scalar parameter, the Fisher information is:
\begin{equation}
\mathcal{I}(\alpha_0) = \sum_{t=0}^T \sum_{i < j} \biggl(\frac{\partial P_{ij}(t)}{\partial \alpha_0}\biggr)^2 \cdot \frac{1}{P_{ij}(t)(1 - P_{ij}(t))} = 4 \sum_{t=0}^T t^2 \sum_{i < j} \frac{P_{ij}(t)}{1 - P_{ij}(t)}
\end{equation}
For the scaling claim: when $P_{ij}(t) \in [\epsilon, 1 - \epsilon]$ for some $\epsilon > 0$, the inner sum $\sum_{i < j} P_{ij} / (1 - P_{ij})$ is $\Theta(n^2)$ (there are $\binom{n}{2}$ terms, each bounded away from zero and infinity).
The outer sum $\sum_{t=0}^T t^2 = T(T+1)(2T+1)/6 = \Theta(T^3)$.
Therefore $\mathcal{I}(\alpha_0) = \Theta(T^3 \cdot n^2)$, giving a Cram\'er--Rao lower bound $\mathrm{Var}(\hat{\alpha}_0) \ge \mathcal{I}^{-1} = \Omega(1/(T^3 n^2))$. \qed

\section{Vertical Projection onto the Fiber}\label{app:vertical-projection}

We derive the formula for projecting a tangent vector $Z \in T_X \RR_*^{n \times d}$ onto the vertical subspace $\mathcal{V}_X = \{X\Omega : \Omega \in \so(d)\}$.

The vertical component $Z^{\mathcal{V}} = X\Omega^*$ is defined by the orthogonality condition $Z - X\Omega^* \in \mathcal{H}_X$ (the remainder is horizontal).
By \cref{prop:horizontal}, a vector $W$ at $X$ is horizontal if and only if $X^\top W$ is symmetric.
Applying this to $W = Z - X\Omega^*$:
\begin{equation}
X^\top (Z - X\Omega^*) = X^\top Z - G\Omega^* \quad \text{must be symmetric}
\end{equation}
where $G = X^\top X$ is positive definite.
Decomposing $X^\top Z$ into symmetric and skew-symmetric parts, $X^\top Z = \sym(X^\top Z) + \skewop(X^\top Z)$, the condition becomes:
$\skewop(G\Omega^*) = \skewop(X^\top Z)$.

Since $\Omega^*$ is skew-symmetric and $G$ is symmetric, the skew-symmetric part of $G\Omega^*$ is:
\begin{equation}
\skewop(G\Omega^*) = \frac{1}{2} (G\Omega^* - (G\Omega^*)^\top) = \frac{1}{2} (G\Omega^* + \Omega^* G)
\end{equation}
where the last step uses $\Omega^{*\top} = -\Omega^*$ and $G^\top = G$.
Equating:
\begin{equation}
G\Omega^* + \Omega^* G = 2 \skewop(X^\top Z)
\end{equation}

This is a Lyapunov equation in $\Omega^*$.
Since $G$ is positive definite, all sums $\lambda_\iota + \lambda_\gamma > 0$, so the equation has a unique solution.
In the eigenbasis of $G = \diag(\lambda_1, \ldots, \lambda_d)$, the solution is elementwise:
\begin{equation}
\Omega^*_{\iota\gamma} = \frac{2 \skewop(X^\top Z)_{\iota\gamma}}{\lambda_\iota + \lambda_\gamma}
\end{equation}

For the application in \cref{prop:curvature-criterion}, $Z = SX$ with $S$ skew-symmetric.
Then $X^\top Z = X^\top S X$, which is already skew-symmetric (since $(X^\top S X)^\top = X^\top S^\top X = -X^\top S X$), so $\skewop(X^\top S X) = X^\top S X$ and the Lyapunov equation becomes $G\Omega^* + \Omega^* G = 2 X^\top S X$.

\section{Norm of the Vertical Bracket Component}\label{app:vertical-norm}

We derive the explicit formula for $\|[\bar{\xi}_1, \bar{\xi}_2]^{\mathcal{V}}\|^2$ stated in \cref{prop:curvature-criterion}.

The vertical component is $X\Omega^*$ with $\Omega^*_{\iota\gamma} = 2(X^\top S X)_{\iota\gamma} / (\lambda_\iota + \lambda_\gamma)$ where $S = -[M_1, M_2]$ (\cref{app:vertical-projection}).
Its squared norm in the ambient Euclidean metric is:
\begin{equation}
\|X\Omega^*\|^2 = \tr((X\Omega^*)^\top X\Omega^*) = \tr(\Omega^{*\top} G \Omega^*)
\end{equation}

Working in the eigenbasis of $G = \diag(\lambda_1, \ldots, \lambda_d)$:
\begin{equation}
\tr(\Omega^{*\top} G \Omega^*) = \sum_{\iota, \gamma} \lambda_\iota (\Omega^*_{\iota\gamma})^2
\end{equation}
(using $\tr(A^\top B A) = \sum_a b_a \sum_b A_{ab}^2$ when $B = \diag(b_1, \ldots, b_d)$; diagonal terms vanish since $\Omega^*_{\iota\iota} = 0$).
Collecting the pair $(\iota, \gamma)$ and $(\gamma, \iota)$ for $\iota < \gamma$ and using skew-symmetry $(\Omega^*_{\gamma\iota})^2 = (\Omega^*_{\iota\gamma})^2$:
\begin{equation}
= \sum_{\iota < \gamma} (\lambda_\iota + \lambda_\gamma)(\Omega^*_{\iota\gamma})^2
\end{equation}

Substituting $\Omega^*_{\iota\gamma} = 2(X^\top S X)_{\iota\gamma} / (\lambda_\iota + \lambda_\gamma)$:
\begin{equation}
= \sum_{\iota < \gamma} (\lambda_\iota + \lambda_\gamma) \cdot \frac{4(X^\top S X)_{\iota\gamma}^2}{(\lambda_\iota + \lambda_\gamma)^2} = 4 \sum_{\iota < \gamma} \frac{(X^\top S X)_{\iota\gamma}^2}{\lambda_\iota + \lambda_\gamma}
\end{equation}

It remains to express $(X^\top S X)_{\iota\gamma}$ in terms of $(U^\top [M_1, M_2] U)_{\iota\gamma}$.
Writing $X = U \Lambda^{1/2}$ in the canonical gauge (where $U$ has orthonormal columns spanning $\col(X)$ and $\Lambda = \diag(\lambda_1, \ldots, \lambda_d)$):
$X^\top S X = \Lambda^{1/2} U^\top S U \Lambda^{1/2}$
so $(X^\top S X)_{\iota\gamma} = \sqrt{\lambda_\iota} (U^\top S U)_{\iota\gamma} \sqrt{\lambda_\gamma}$.
Since $S = -[M_1, M_2]$:
$(X^\top S X)_{\iota\gamma}^2 = \lambda_\iota \lambda_\gamma (U^\top [M_1, M_2] U)_{\iota\gamma}^2$.

Substituting:
\begin{equation}
\|X\Omega^*\|^2 = 4 \sum_{\iota < \gamma} \frac{\lambda_\iota \lambda_\gamma}{\lambda_\iota + \lambda_\gamma} [(U^\top [M_1, M_2] U)_{\iota\gamma}]^2
\end{equation}

This is the formula stated in \cref{prop:curvature-criterion}.
The $1 / (\lambda_\iota + \lambda_\gamma)$ weighting shows that curvature is amplified when both eigenvalues in a pair are small, matching the connection coefficient structure of the fiber bundle.

\section{Extension to Directed Graphs}\label{app:directed}

This appendix is a sketch intended to indicate how the framework adapts to directed graphs. Details are omitted.

For directed graphs, each node has source position $g_i$ and target position $r_i$.
The probability matrix is $P = GR^\top$ (not symmetric).

The gauge group becomes $(G, R) \sim (GQ, RQ)$ for $Q \in O(d)$, and the invisible dynamics are $\dot{G} = GA$, $\dot{R} = RA$ with $A \in \so(d)$.

The gauge-consistent architecture generalizes to:
\begin{equation}
\dot{G} = N_G(P) G, \quad \dot{R} = N_R(P) R
\end{equation}
with appropriate symmetry constraints on $(N_G, N_R)$ ensuring that the induced dynamics on $P = GR^\top$ are well-defined.
The fiber bundle framework carries over: the total space is pairs $(G, R) \in \RR_*^{n \times d} \times \RR_*^{n \times d}$, the structure group is $O(d)$ acting diagonally, and the connection, curvature, and holonomy theory applies with the same Lyapunov-type structures, though the asymmetry of $P$ introduces additional subtleties in the spectral decomposition.

\end{document}